\documentclass{llncs}
 
\usepackage{amsmath,amssymb}
\usepackage{mathtools}
\usepackage{enumerate}
\usepackage[usenames,dvipsnames]{color}

\usepackage{enumitem}
\usepackage{graphicx}
\usepackage{transparent}
\usepackage{color}

\graphicspath{{img/}}

\newcommand{\R}{\mathbb{R}}

\newcommand{\cone}[2]{\Phi_{#1,#2}}
\newcommand{\conestrat}[3]{\cone{#1}{#2}^{#3}}
\newcommand{\define}[1]{\emph{#1}}
\newcommand{\cayley}[5]{\Phi_{#1,#2}^{#3}({#4},{#5})}

\setenumerate[1]{label=\Roman*.}
\setenumerate[2]{label=\Alph*.}
\setenumerate[3]{label=\roman*.}
\setenumerate[4]{label=\alph*.}

\newtheorem{observation}[theorem]{Observation}

\begin{document}

\title{On Flattenability of Graphs}
\author{Meera Sitharam \thanks{This research was supported in part by the grant  NSF CCF-1117695} \and Joel Willoughby}
\institute{University of Florida}

\maketitle

\begin{abstract}
We consider a generalization of the concept of $d$-flattenability of 
graphs - introduced for the $l_2$ norm by Belk and Connelly - to general $l_p$ norms, with integer $P$, $1 \le p < \infty$, though many of our results work for $l_\infty$ as well. 
The following results are shown for graphs $G$, using notions of 
genericity, rigidity, and generic $d$-dimensional rigidity matroid introduced 
by Kitson for frameworks in general $l_p$ norms,  as well as the cones of
vectors of pairwise $l_p^p$ distances of a  finite 
point configuration in $d$-dimensional, $l_p$ space:  
(i) $d$-flattenability of a graph $G$ is equivalent to the convexity of 
$d$-dimensional, inherent Cayley configurations spaces for $G$, 
a concept introduced by the first author; (ii) $d$-flattenability and 
convexity of Cayley configuration spaces over specified non-edges of 
a $d$-dimensional framework are not 
generic properties of frameworks (in arbitrary dimension); (iii) $d$-flattenability
of $G$ is equivalent to all of $G$'s generic frameworks being $d$-flattenable; 
(iv) existence of one generic $d$-flattenable framework for $G$ is equivalent to 
the independence of the edges of $G$, a 
generic property of frameworks; (v) the rank of $G$ equals the dimension of 
the projection of the $d$-dimensional stratum of the $l_p^p$ distance cone.
We give stronger results for specific norms for $d = 2$: we show
that (vi) 2-flattenable graphs for the $l_1$-norm (and $l_\infty$-norm) 
are a larger class than 2-flattenable graphs for Euclidean $l_2$-norm case 
and finally (vii) prove further results towards characterizing 
2-flattenability in the $l_1$-norm. A number of conjectures and open 
problems are posed.





\end{abstract}

\section{Introduction, Preliminaries, Contributions}
\label{intro}
A \define{realization or framework of a graph} $G = (V,E)$ 
\define{under norm} $|| \cdot ||$ is  
is an assignment $r: V\to \R^m$  
of points in the corresponding normed vector space 
$\R^m$. 
A \define{linkage} $(G, \delta_G)$ is a graph $G=(V,E)$ together with an assignment $\delta_G: E \to \R$ of positive real assignments of lengths to the edges of $G$. 
A \define{realization of a linkage} $(G,\delta^G)$ 
\define{in $d$ dimensional $||.||$-normed space} is 
an assignment $r: V \to \R^d$,   such that 
$\forall (v,w) \in E, ||r(v) - r(w)|| = \delta^G_{vw}$. 
A realization \define{under norm $||.||$} is a realization in $d$ dimensional
$||.||$-normed space, for some dimension $d$. In this paper, we are concerned with standard $l_p$ norms. By \define{general $l_p$ norms}, we mean norms with integer $p$, $1 \le p < \infty$. However, many results of this paper hold for $l_\infty$ as well.
While under the $l_2$ norm a realization $r$
of intrinsic dimension $d$ - i.e., whose points lie on a $d$-dimensional subspace of some higher
$d'$-dimensional space - is linearly isometric to a
$d$-dimensional realization, this is not the case for other $l_p$ norms unless the subspace is axis parallel, i.e, a
 coordinate subspace or a translated (affine) subspace..
Hence  the dimension of such a realization $r$ under general norms is considered to be $d'$ rather than $d$.
A graph $G$ is \define{$d$-flattenable} if 
for every realization $r$ of $G$ 
under norm $||.||$, the linkage $(G,\delta^G)$ where 
$\delta^G_{vw} := ||r(v) -r(w)||$ 
This an illustration of a 2-flattenable graph that does not refer to realizations of intrinsic dimension 2 in some higher dimensional space.
also has a realization in the $d$-dimensional $||.||$-normed space. This definition does not imply that there is a continuous path of realizations starting from a realization of $(G, \delta_G)$ in some higher dimension to the realization in $d$-dimensions, nor does it refer to realizations of intrinsic dimension $d$ in some higher dimensional space. For a clarification of the latter, see the example in the Proof of Theorem \ref{thm:2sumK4} where we give a realization of a graph on a 2-dimensional subspace of $\R^3$, this does not imply that the graph is 2-flattenable. This particular graph turns out to be 2-flattenable as every $l_1$ realizable linkage of it can be realized in $\R^2$.

This concept was 
first introduced in \cite{Belk:2007} for the Euclidean or $l_2$ norm. However they called it "$d$-realizability,"
which can be confused with the realizability of a given linkage in $d$-dimensions. This is one of reasons we introduced the term: {\em flattenability.}

The term flattening has also been used  by Matousek  \cite{Matousek2004}
in the context of non-isometric embeddings (with low distortion via
Johnson-Lindenstrauss lemma in $l_2$ \cite{JohnsonLindenstrauss},  impossibility 
of low distortion in $l_1$ \cite{Brinkman2005}, etc).
 Our paper admits arbitrary distortions of non-edge lengths, but forces edge lengths to remain undistorted.

A \define{minor} of $G$ is any graph $G'$ that can be obtained from $G$ from a series of edge-contractions or edge-deletions. If a property of $G$ remains consistent under the operation of taking minors, that property is \define{minor-closed}. A useful result due to \cite{RobertsonAndSeymour} is that if a property is minor-closed, then there is a finite set of \define{forbidden minor} $\mathcal{F}$ such that if $G$ has any element of $\mathcal{F}$ as a minor, then $G$ does not have that property.

Immediately by definition, $d$-flattenability is a minor-closed property
under any norm. A full characterization for 3-flattenable graphs was given 
for the Euclidean or  $l_2$ norm  by \cite{Belk:2007}.

    This paper gives basic results illustrating how
    $d$-flattenability for general norms is a natural link
    between
    combinatorial rigidity and 
    configuration spaces of frameworks on the one hand, and    
    coordinate shadows
    (projections) of the faces of the cone -- consisting of vectors of 
    pairwise $l_p^p$-
    distances of $n$-point configurations (see Figure \ref{fig:vector_example}) -- on the other hand (see Figure \ref{fig:cone_dflat}). We define the \define{$l_\infty^\infty$ cone} to be the limit of the $l_p^p$ cones
    as $p\rightarrow \infty$. This definition permits some of our results to
    hold for the $l_\infty$ norm as well. 
    Thus, via $d$-flattenability, 
    graph minors and topological embeddings, 
    as well as combinatorial rigidity tools 
    can now be used to understand the structure of these 
    cone faces that play a crucial 
    role in convex and semidefinite programming, spectral graph theory
    and metric space embedding \cite{Steurer:SoS}.The latter techniques are used widely in  
    approximation of  optimal solutions to {\em NP}-hard combinatorial problems and in complexity theory, where
    in particular, non-Euclidean norms such as 
    $l_1$ and $l_\infty$ play a crucial role \cite{Khot05,Trevisan}. 
    Thus $d$-flattenability is a nexus connecting 
    diverse techniques and applications.

\begin{figure} 
    \centering\tiny
    \def\svgwidth{0.2\linewidth} 
    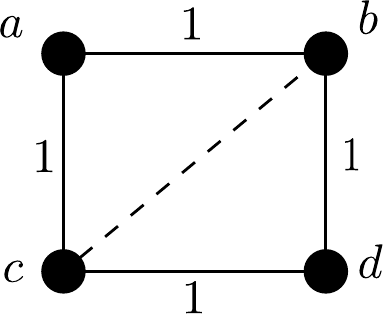 
    \caption{Example of a linkage. The corresponding {\em pairwise
    distance vector} for this graph is given by: $\delta =
    (1,1,1,1,||a-d||,||b-c||)^T$. The Cayley configuration space on the
    non-edge $ad$ can take any distance in the range $[0,2]$. The ordering of
    pairs as coordinate positions in the vector is arbitrary, but fixed by
    convention.}
\label{fig:vector_example}
\end{figure}

\bigskip\noindent
    In the remainder of this section we give preliminary definitions, 
    state the paper's contributions and organization, and provide a brief listing of
    related work on the above topics in Section \ref{related}.

    In \cite{SiGa:2010} one of the authors introduced an alternative 
    perspective on the configuration or realization space for a 
    given linkage $(G,\delta^G)$, defining the   
    \define{$d$-dimensional Cayley configuration space} over some set of 
    non-edges, $F$, of $G$ 
    under the $l_2^2$ norm.
    This Cayley configuration space is
    denoted $\Phi^{d}_{F,l_2}(G,\delta^G)$, and is 
    the set of vectors $\delta^F$ of  Euclidean lengths attained
    by the non-edges $F$ over all the realizations of the linkage $(G,\delta^G)$. 
    This same space is also sometimes referred to as the {\em Cayley configuration space of any
    realization or framework} $(G,r)$ whose edge lengths are $\delta^G$.
    The definition readily extends to arbitrary norms.
    In
    \cite{SiGa:2010}, it was shown that for the $l_2$ norm, 
    $d$-flattenability of a graph
    $G$ implies $G$ has a \define{$d$-dimensional,
    inherent convex Cayley} configuration
    space, i.e., for all partitions of $G = H\cup F$, and all length vectors
    $\delta^H$ for the edges of $H$, $\Phi^{d}_{F,l_2}(H,\delta^H)$ is a convex set (see Figure \ref{fig:convexcayley}). 
     This property
    was then used towards highly efficient 
    atlasing of molecular configuration spaces 
    \cite{Ozkan:2011}, 
    compared and hybridized with standard monte carlo methods in \cite{MonteCarlo},
    \cite{OzkanJacobian},   
    with multiple applications demonstrated in \cite{Ozkan:2011,WuVirus}. 
    Our first result in Section \ref{joel1} 
    shows the converse of the above result and generalizes both directions to general
    $l_p$ norms, leading to our first main result:

    \begin{itemize}
        \item
            For $l_p$ norms,
    $G$ is $d$-flattenable if and only if $G$ has a \define{$d$-dimensional,
    inherent convex Cayley} configuration space.
    As a direct corollary, it follows that 
    both properties are minor-closed for general $l_p$ norms. 
    \end{itemize}

    For the next set of results given in Section \ref{meera}, we refer the reader to combinatorial
    rigidity preliminaries in \cite{Graver}, defined 
    for the Euclidean or $l_2$ normed space.
    The \define{$d$-dimensional rigidity matrix} of a graph  $G = (V,E)$, 
    denoted $R(G)$,
is a matrix of indeterminates $r_1(v), r_2(v), \ldots r_d(v)$ for $v\in V$.
These represent the coordinate position $r(v) \in \R^d$ of  
the point corresponding to a 
vertex $v\in V$ in an arbitrary realization or framework $r$ of $G$.
The matrix has one row for each edge
each vertex $v \in V$.
The row corresponding to $e = (u,v) \in E$ 
represents the {\em bar} 
from $r(u)$ to $r(v)$ and 
has $d$ non-zero entries
$r(u)-r(v)$ (resp. $r(v)-r(u)$), in the $d$ columns corresponding to $u$
(resp. $v$).
An instantiation of $R(G)$ to a particular framework
is called the rigidity matrix of that framework.
A \define{regular} or \define{generic} framework $(G,r)$ (with respect to infinitesimal rigidity), is one whose corresponding
instantiation of $R(G)$ has maximal rank over all instantiations.

A subset of edges of a graph $G$  is  said to be 
\define{independent} if 
the corresponding set of rows of $R(G)$ are generically independent. 
The maximal independent set  yields the rank of $G$ in the \define{$d$-dimensional rigidity matroid} (independent sets of edges of 
the complete graph).
The graph (resp. generic framework) is 
\define{(resp. infinitesimally) rigid} if the number of generically independent rows or 
the rank of $R(G)$  is maximal, i.e., $d|V| - {d+1 \choose 2}$, where 
${d+1 \choose 2}$ is the number of Euclidean isometries
in $\R^d$ \cite{Graver}.

For frameworks in polyhedral norms (including the $l_p$ norms), 
    Kitson \cite{Kitson:2014} 
    has defined properties such as {\em well-positioned, regular} analogous  to the above,
    which have been used 
    to show 
(infinitesimal) rigidity 
to be a generic property of frameworks.

We refer the reader to Kitson's paper for a precise definition. Intuitively, a well-positioned $d-$dimensional framework under norm $||.||$ is one in whose $d$-dimensional neighborhood in $||.||$-normed space the pairwise distances between points can be expressed in polynomial form.  

\begin{itemize}
    \item
    For general 
    $l_p$ frameworks in arbitrary
    dimension, $d$-flattenability of a graph $G$ is equivalent to
    all  generic frameworks of $G$ being $d$-flattenable. 
\item
    However, already for the Euclidean or $l_2$ case, 
    $d$-flattenability
    is not a generic property of  frameworks (in arbitrary dimension),
    and neither is the convexity of 
    Cayley configuration
    spaces over specified non-edges of a $d$-dimensional framework. 
    The latter uses minimal, 1-dof Henneberg-I frameworks for $d=2$
    constructed in 
    \cite{SiWa:2012,SiWa:2013}. 
\item
    The {\sl existence} 
    of a generic $d$-flattenable framework (in arbitrary dimension) 
    is equivalent to independence of the rows of the generic $d$-dimensional 
rigidity matrix of its graph 
-  we use the genericity concepts developed by Kitson \cite{Kitson:2014}
for $l_p$ norms.
    \end{itemize}

\begin{figure} 
    \centering\tiny
    \def\svgwidth{0.3\linewidth} 
    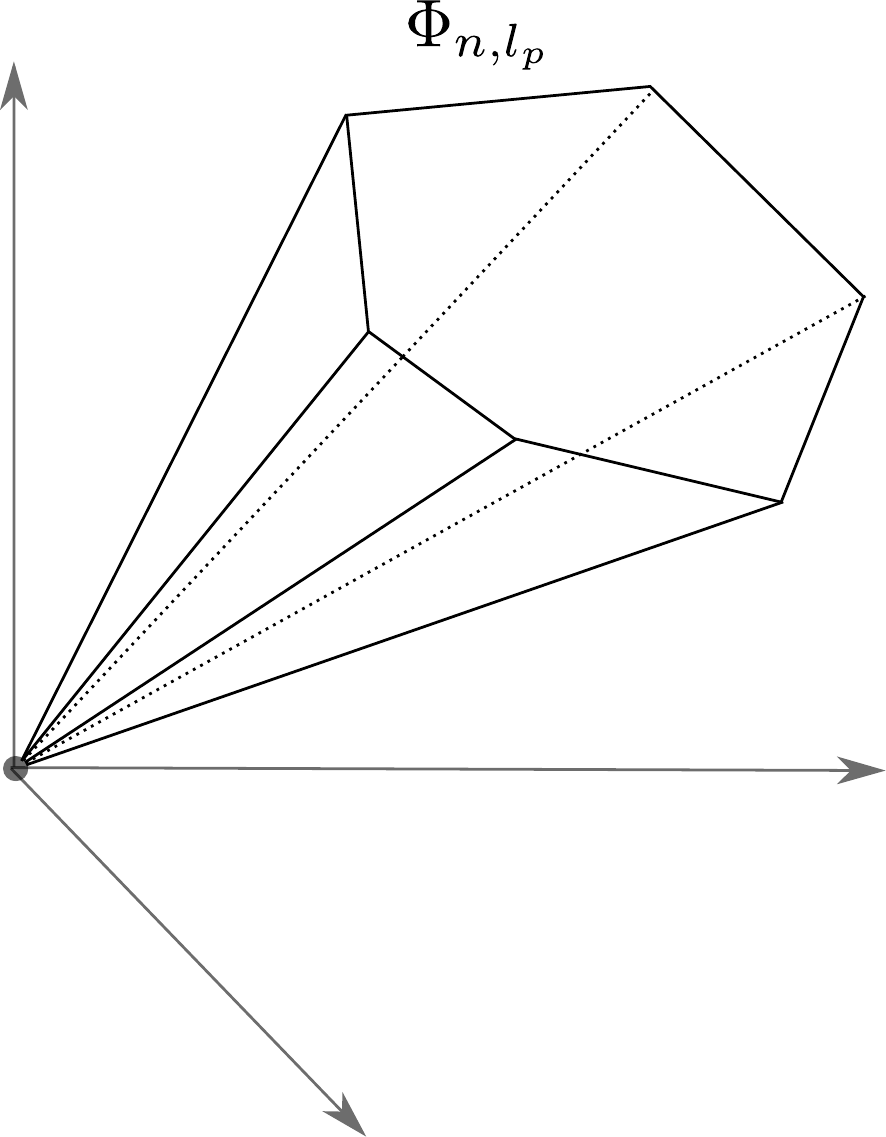 
    \def\svgwidth{0.3\linewidth} 
    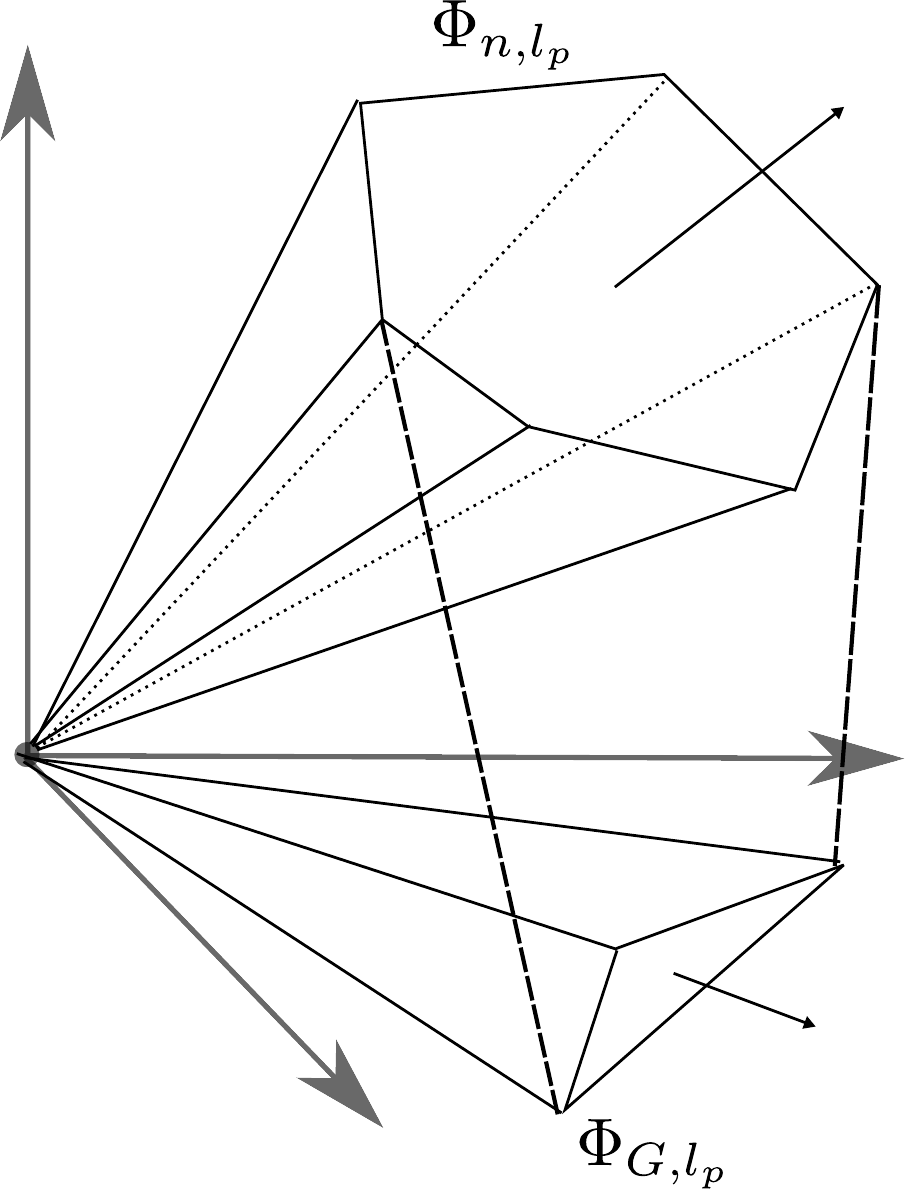
    \def\svgwidth{0.3\linewidth} 
    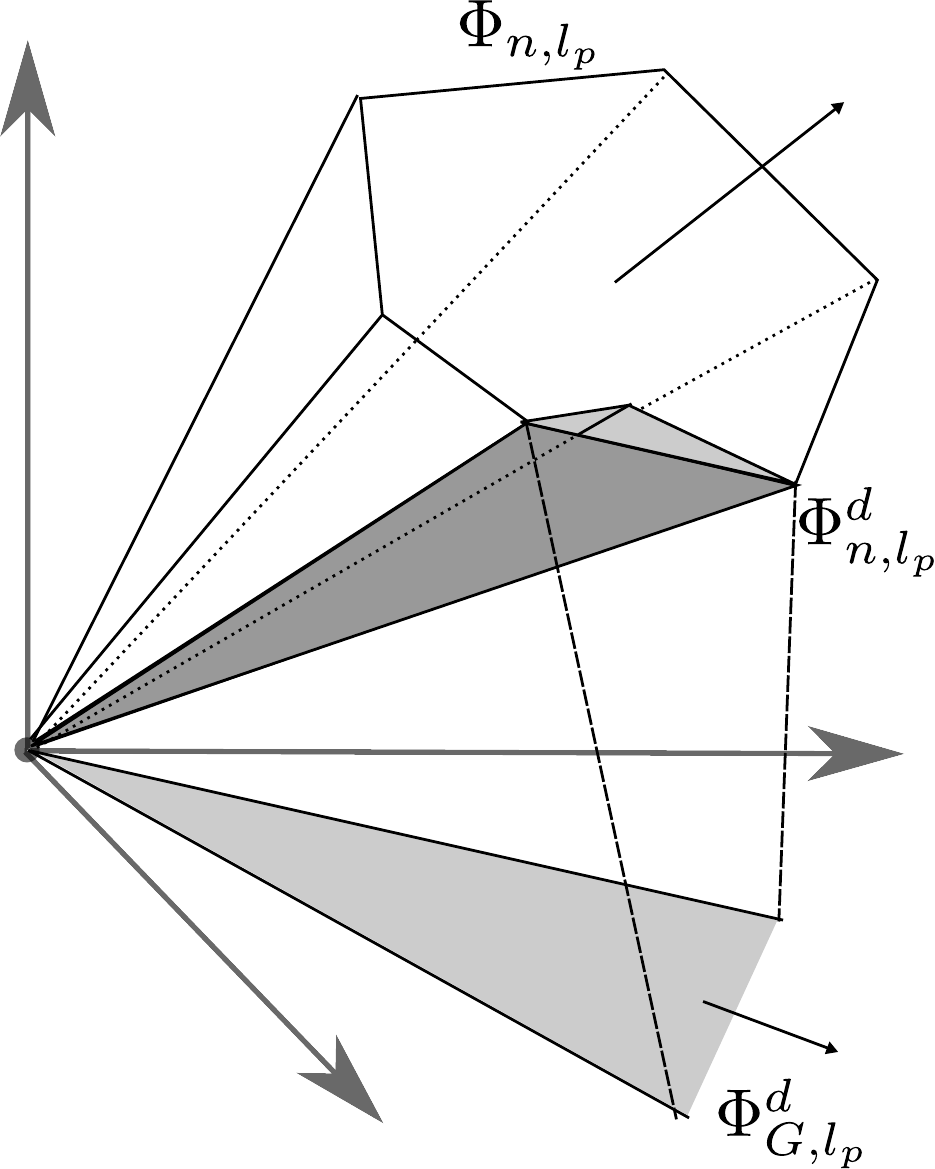
    \caption{Visualizing operations common to our proofs. On the left we have the cone of realizable distance vectors under $l_p$. \emph{It is shown here as a polytope, but in general that is not the case; these are not rigorous figures -- their purpose is intuitive visualization}. The cone lives in ${n \choose 2}$-dimensional space where each dimension is a pairwise distance among $n$ points. In the middle is a projection onto the edges of some graph. This will yield a lower dimensional object (unless $G$ is complete). On the right, a $d$-dimensional stratum is highlighted and lines show the projection onto coordinates representing edges of a graph. In general this stratum is not just a single face. Note that this projection is equal to the projection as the whole cone (middle) iff $G$ is $d$-flattenable.}
\label{fig:cone_dflat}
\end{figure}

    The next result, also in Section \ref{meera}  concerns the cone $\cone{n}{l_p}$    
 consisting of vectors $\delta_r$ of 
    pairwise $l_p^p$-distances of $n$-point configurations $r$. 
 (A proof that this set is a cone 
    can be found in \cite{Ball},  which also applies to infinite dimensional settings).

    The 
\define{$d$-dimensional stratum} of this cone consists of 
pairwise distance vectors of $d$-dimensional point configurations and is denoted
$\conestrat{n}{l_p}{d}$.
The projection or shadow of this cone (resp. stratum) 
on a subset of coordinates 
    i.e., pairs corresponding to the edges of a graph $G$
    is denoted $\cone{G}{l_p}$ (resp. $\conestrat{G}{l_p}{d}$).
    This projection is the set of 
    realizable  edge-length
    vectors $\delta^G$ of linkages $(G,\delta^G)$ in $l_p^p$ (resp. in $d$-dimensions) (See Figure \ref{fig:cone_dflat}).

Notice that  $\cone{n}{l_p}$ is the same as $\cone{K_n}{l_p}$, where
$K_n$ is the complete graph on $n$ vertices.
The \define{$l_p$-flattening dimension} of a graph $G$ (resp. class $C$ of graphs)
is the minimum dimension $d$ for which $G$ (resp. all graphs in $C$) are flattenable in
$l_p$.
Let $n_p$ be the flattening dimension of $K_n$.  
It is not hard to show \cite{DezaAndLaurent} that in fact  $n_p \le \R^{{n\choose2}}$  
   (using this finite dimensionality,  a slight simplification of Ball's proof of convexity of  $\cone{n}{l_p}$  is presented for completeness in Section \ref{joel1}).  For the Euclidean or $l_2$ case, a further result 
of Barvinok \cite{Barvinok1995} shows that the flattening dimension of 
any graph $G = (V,E)$ (although he did not use this
terminology),  is 
at most $O(\sqrt(|E|))$.
    Notice  additionally that $\Phi_{F,l_p}^{d}(G,\delta^G)$, namely 
    the $d$-dimensional Cayley configuration space of a 
    linkage $(G,\delta^G)$ in $l_p$
    is the coordinate shadow of the  
    $(G,\delta^G)$-fiber of $\conestrat{G\cup F}{l_p}{d}$, i.e. all 
    linkages $(G \cup F, \delta_{G \cup F})$ that have $\delta_G$
    assigned to the edges of $G$,
    on the coordinate set $F$ (see Figure \ref{fig:convexcayley}). In this paper, we show the following:

    \begin{figure} 
    \centering\tiny
    \def\svgwidth{0.29\linewidth} 
    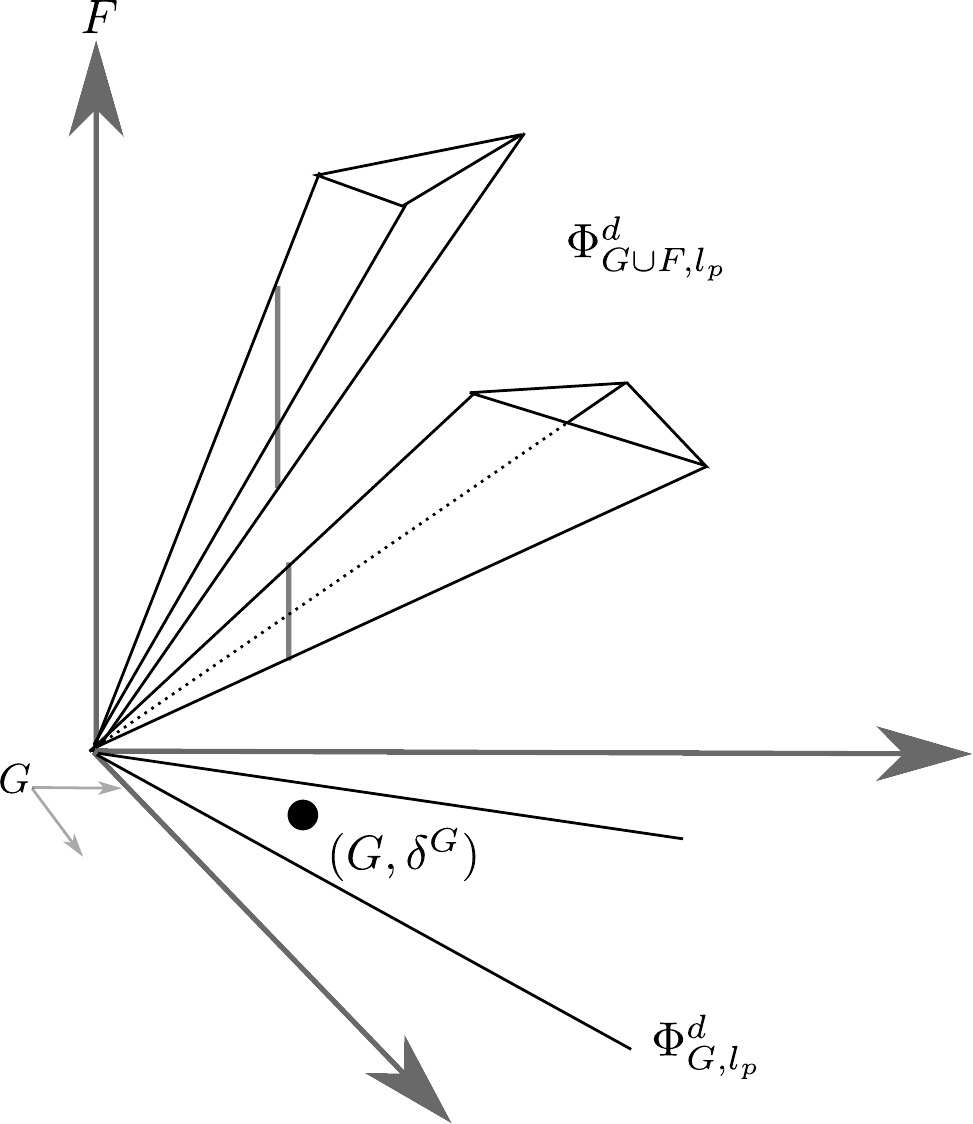 
    \def\svgwidth{0.3\linewidth} 
    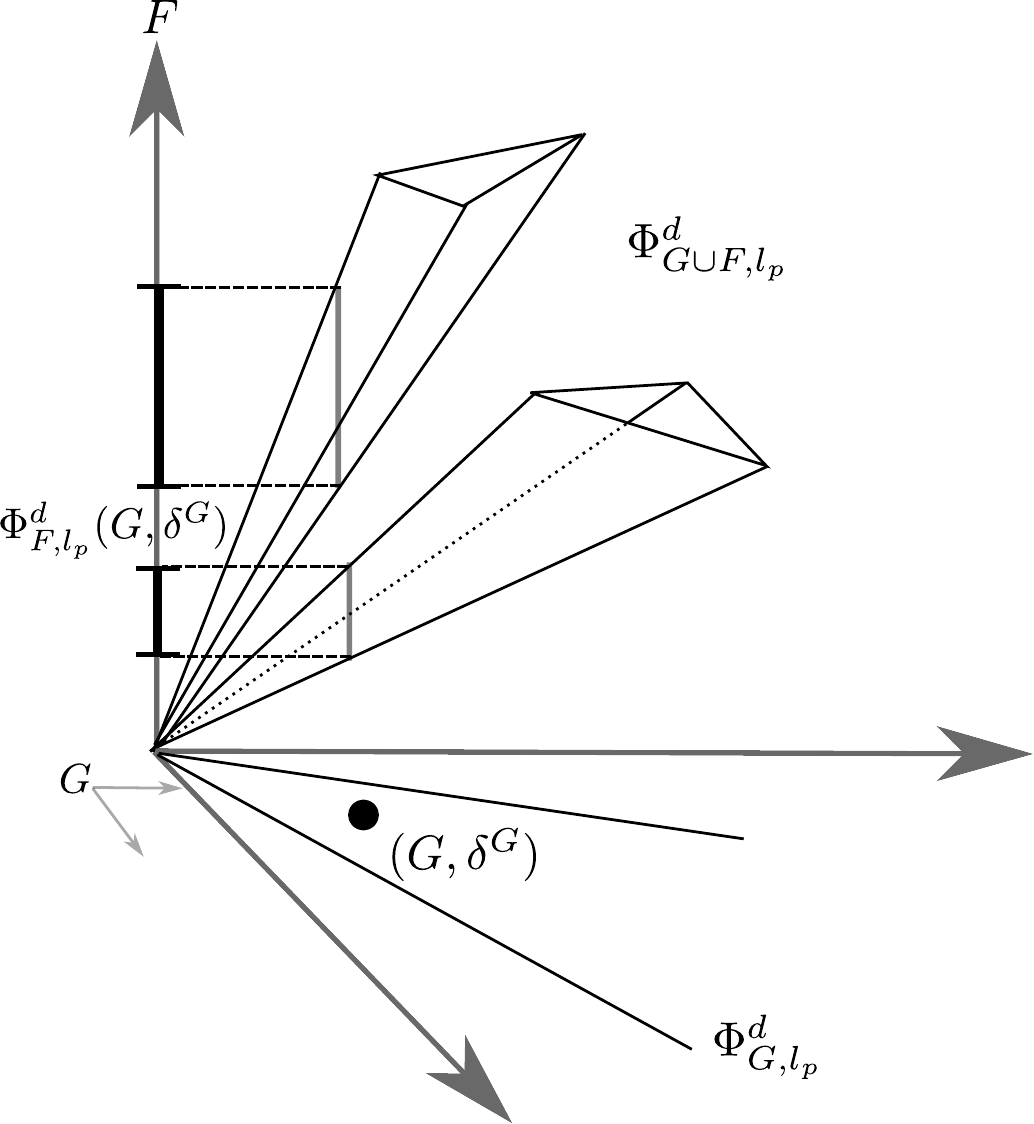
    \def\svgwidth{0.28\linewidth} 
    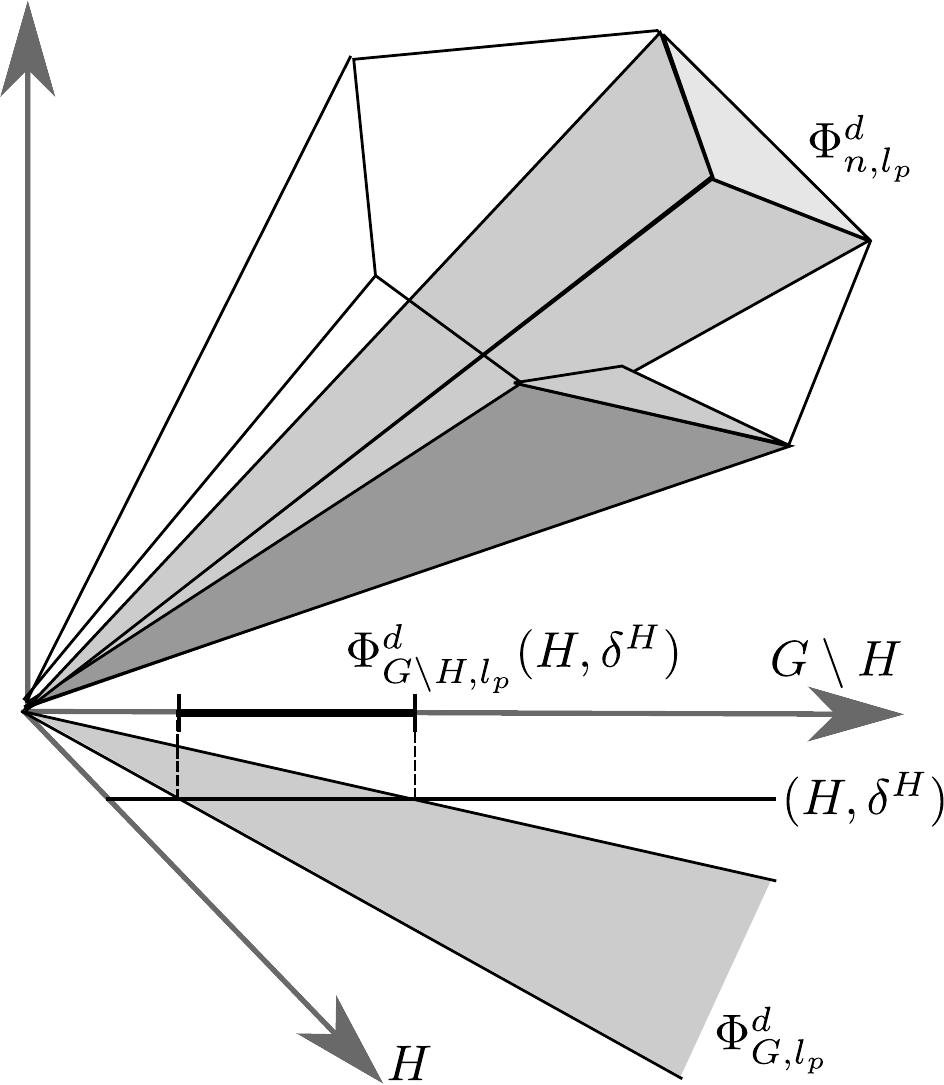
    \caption{This is an example of $\conestrat{G \cup F}{l_p}{d}$ that is not convex. The linkage $(G,\delta^G)$ and its fiber in $\conestrat{G \cup F}{l_p}{d}$ are shown on the left. Note that the fiber is not convex. In the middle, this fiber is then projected onto the remaining edges of $G \cup F$ to form $\cayley{F}{l_p}{d}{G}{ \delta^G}$. Note that it is not convex either. On the right, $\conestrat{n}{l_p}{d}$ is projected onto the edges of some $d$-flattenable $G$ (note that this is the same as projection of $\cone{n}{l_p}$). The inherent Cayley configuration space corresponding to some subgraph $G \setminus H$ of $G$ is then shown projected onto the edges of $G \setminus H$. This projection is convex.}
\label{fig:convexcayley}
\end{figure}

    \begin{itemize}
        \item
    Consider 
    the coordinate shadow (or projection) of any
    neighborhood in the stratum
$\conestrat{n}{l_p}{d}$
     onto
    the edges of an $n$-vertex graph $G$. The dimension 
    of this coordinate shadow
    equals the rank of $G$ (size of maximal  independent set) in the generic 
    $d$-dimensional rigidity matroid \cite{Kitson:2014} in $l_p$. 
    \end{itemize}

    In Section \ref{joel2}, we give stronger results for specific norms for $d=2$:
    \begin{itemize}
        \item
     The class of 2-flattenable graphs for the $l_1$-norm 
     (and $l_\infty$-norm)  strictly contains 
    the class of 2-flattenable graphs for the Euclidean $l_2$-norm case, 
    (the latter
    being
    the partial 2-tree graphs that avoid the $K_4$ minor). 
    In particular, $K_4$ is 2-flattenable in $l_1$.  
    Graphs with {\sl Banana} graphs as minors, 
    however, are not 2-flattenable. We also consider 
    other graphs such as the {\sl 4-wheel} and 
    the {\sl doublet and $K_{3,3}$} towards obtaining a 
    forbidden-minor characterization of 2-flattenability in the $l_1$-norm.
    \end{itemize}

    Finally, in Section \ref{open}, a number of conjectures and 
    open problems are posed.

\subsection{Related Results}
\label{related}
The structure of the cone $\cone{n}{l_p}$, its strata and faces are well-studied.
The fact that the object is a cone even in the infinite dimensional case is a 
useful observation by Ball \cite{Ball}.
For $l_2$ this is called the \define{Euclidean distance matrix (EDM)} 
cone \cite{Dattorro2011,Tarazaga,Pataki}, which is  
a simple, linear transformation of the cone of positive semidefinite matrices,
a fact first observed by Schoenberg \cite{Schoenberg35}.
Consequently understanding its structure is important in 
semidefinite programming relaxations and the so-called sums of squares
method with numerous applications \cite{Parrilo03,Parillo2004,Steurer2011}. 
Connections between combinatorial rigidity and the structure of the EDM
have been investigated extensively by Alfakih \cite{Alfakih2000,Alfakih2010}. 
The reader is additionally referred to \cite{DezaAndLaurent} for a
comprehensive survey of key results about the EDM cone, including observations about the face
structure and dimensional strata of the EDM cone.
The $l_1$-cone is often called the \define{cut cone}, whose extreme rays correspond to 1-dimensional realizations and 
characteristic vectors of cuts in a complete graph.  The cone has been
studied by \cite{Avis1991,DezaAndLaurent,Witsenhausen1986} and plays an important
role in metric space embeddings used in the study of (non-)approximability in polynomial time,  of
NP-hard optimization problems, including ramifications of the unique games conjecture
\cite{Matousek2004}, \cite{Khot05}. Kitson's recent work
\cite{Kitson:2014} has shown that many of the results in combinatorial
rigidity for the Euclidean or $l_2$ norm case have parallels in the case of
general polyhedral norms, including the $l_p$ norms.

    \section{$l_p$: Flattenability and Inherent Convex Cayley Configuration Space} \label{joel1}

In this section, we improve in the work done in \cite{SiGa:2010} in relating $d$-flattenability to convex inherent Cayley configuration spaces of a given graph. The results of this section hold
for $p = \infty$ as well, appealing to the definition of
$l_\infty^\infty$ cone as the limit of the $l_p^p$ cones as $p \rightarrow \infty$. Our main result of this section  follows.

\begin{theorem}\label{thm:covexcayley_d-flattenable}
  For any $l_p$ norm, a graph $G$ is $d$-flattenable if and only if $G$ admits convex inherent $d$-dimensional Cayley configuration spaces for each of its subgraphs.
\end{theorem}

This ``only if'' direction of this statement was shown in \cite{SiGa:2010} for the $l_2$ norm. The argument only required the fact that the cone of squared distance vectors is convex. Hence, we can use the same proof if we can show $\cone{n}{l_p}$ is convex. The proof of the ``if'' direction requires that the cone is the convex hull of $l_p^p$ distance vectors in any dimension $d$. 


\begin{proposition}
\label{1-d Convex Hull}
  $\cone{n}{l_p}$ for general $l_p$ is contained in the convex hull of the $l_p^p$ distance vectors of the 1-dimensional $n$-point configurations in $\R$.
\end{proposition}

\begin{proof}
  Take some $\delta \in \cone{n}{l_p}$. Let $r(1), ... r(n)$ denote some realization of the complete linkage $(K_n,\delta)$. We refer to this as a realization of $\delta$.  So, $r(i) \in \R^k$ for 
  some $k$. It was shown in \cite{Ball} that for any $l_p$-norm, the flattening dimension $n_p \le {n \choose 2}$, so there is a realization in some finite dimension. We have 
  \[
  \delta_{ij} = \| r(i) - r(j)\| _p^p = \displaystyle\sum_{l=1}^k |r_l(i) - r_l(j)|^p
  \]
  where $r_l(i)$ denotes the $l$th coordinate of the $i$th point. Then, if we construct  the matrix $\delta^l$ such that $\delta^l_{ij} = \| r_l(i) - r_l(j) \|_p^p$, then $\delta^l$ is a valid $l_p^p$ distance matrix with an $n$-point configuration in $\R$. This point configuration simply being $r_l(1), ..., r_l(n)$. Also, for any $\alpha > 0$, $\alpha \delta^l$ is a valid $l_p^p$ distance matrix with realization $\alpha^{\frac{1}{p}} r_l(1), ... \alpha^{\frac{1}{p}} r_l(n)$. Finally, $\delta = \sum_l \frac{1}{k} [k\delta^l]$, which is a convex combination of $n$-point configurations in $\R$.

\qed
\end{proof}

A well-known result shows that $\cone{n}{l_p}$ is convex.

\begin{observation}
\label{Omega is convex}
  $\cone{n}{l_p}$ for general $l_p$ is convex
\end{observation}

\begin{proof}
  The proof for this result (see \cite{Ball}) is well known even for the infinite dimensional case. Here we give a simplified proof for finite dimensions for completeness.

  Let $r$ and $s$ be two $n$-point configurations with corresponding distance vectors $\delta_r, \delta_s \in \cone{n}{l_p}$. Assume $r$ and $s$ are realized in some dimension $k$. Let $0 \le \lambda \le 1$ and consider the convex combination $\delta = \lambda \delta_r + (1-\lambda) \delta_s$. We will construct an $n$-point configuration in $2k$ dimensions with $\delta$ as its distance matrix. Note that $\delta_{ij} = \| r(i) - r(j) \|_p^p + \| s(i) - s(j) \|_p^p = \sum_l^k |r_l(i) - r_l(j)|^p + \sum_l^k |s_l(i) - s_l(j)|^p$. Then a realization for $t$ can be found by simply concatenating the coordinates of $r$ and $s$ and scaling them appropriately:
  \begin{align*}
    t = (\lambda^{\frac{1}{p}} r, (1-\lambda)^{\frac{1}{p}} s)   
  \end{align*} 
  It is easy to verify that $t$ is a realization of $\delta$.
\qed
\end{proof}

Proposition \ref{1-d Convex Hull} and Observation \ref{Omega is convex} lead to the following, which is useful to us in proving Theorem \ref{thm:covexcayley_d-flattenable}.

\begin{proposition}\label{thm:convexhull_dimensions}
  $\cone{n}{l_p}$, $1 \le p \le \infty$ is the convex hull of the $l_p^p$ distance vectors of the 1-dimensional, $n$-point configuration vectors in $\R$. 
\end{proposition}

\begin{proof}
  Follows from Proposition \ref{1-d Convex Hull} and Observation \ref{Omega is convex} and the fact that in Proposition \ref{1-d Convex Hull}, the points making up the convex hull are in $\cone{n}{l_p}$.
\qed
\end{proof}

Since from the 1-dimensional vectors, we can construct - as convex combinations - vectors realizable in any arbitrary $d$-dimensions, we get the following Corollary:

\begin{corollary}
  $\cone{n}{l_p}$ is the convex hull of the vectors in $\conestrat{n}{l_p}{d}$ for any $d$ for general $l_p$.
\end{corollary}

The following observation is useful in characterizing $d$-flattenability.

\begin{observation}\label{obs:projection_is_cone}
  If $G$ is $d$-flattenable, then the projection of $\conestrat{n}{l_p}{d}$ onto the edges of $G$ is exactly the projection $\cone{n}{l_p}$ onto the edges of $G$.
\end{observation}

Using these results, we can now prove the ``if'' part of Theorem \ref{thm:covexcayley_d-flattenable}.

\begin{proof}[of Theorem \ref{thm:covexcayley_d-flattenable}]
    Suppose $G$ is $d$-flattenable under some $l_p$-norm. Then, because $\cone{|V|}{l_p}$ is convex, $\cone{G}{l_p}$ is convex. From Observation \ref{obs:projection_is_cone}, $\conestrat{G}{l_p}{d}$ is convex. Given a subgraph $F$ of $G$, if we break $G$ into $H$ and $F$ and fix the values of $E$ corresponding to a linkage $(H,\delta^H)$, we are taking a section of $\conestrat{H \cup F}{l_p}{d}$, which is again convex. This is also exactly the Cayley configuration space $\cayley{F}{l_p}{d}{H}{\delta^H}$. Note that this holds for any partition $H$ and $F$, so $G$ always admits a convex $d$-dimensional Cayley configuration space for each of its subgraphs.

    For the other direction, suppose for a linkage $(G,\delta^G)$, all $d$-dimensional Cayley configurations corresponding to subgraphs of $G$, $\cayley{F}{l_p}{d}{G \setminus F}{\delta^{G \setminus F}}$ are convex. Certainly this holds for the empty subgraph as well. We note that $\cayley{G}{l_p}{d}{\emptyset}{\delta^\emptyset}$ is just $\conestrat{G}{l_p}{d}$ and because it is convex, $\conestrat{G}{l_p}{d}$ is its own convex hull. We also know that the convex hull of $\conestrat{G}{l_p}{d}$ is the projection of the convex hull of $\conestrat{|V|}{l_p}{d}$. By Proposition \ref{thm:convexhull_dimensions} and its Corollary, we know this to be the entire cone $\cone{|V|}{l_p}$. Thus, $\conestrat{G}{l_p}{d}$ = $\cone{G}{l_p}$. Hence, $G$ is $d$-flattenable.
\qed
\end{proof}

This result provides a nice link between $d$-flattenability and convex Cayley configuration spaces. It leads to the following tools.

\begin{corollary}
  Having a $d$-dimensional convex Cayley configuration space on all subgraphs is a minor-closed property.
\end{corollary}

Another immediate result is that $d$-flattenability and convex Cayley configuration spaces have the same forbidden minor characterizations for given $d$ under the same $l_p$-norm. This gives us a nice tool when trying to find forbidden minors for other $l_p$ norms:

\begin{observation}\label{obs:not_convex_not_flattenable}
  If for some assignment of distances $l$ to some edges $E$ of of $G$ leads to a non-convex $\cayley{F}{l_p}{d}{G}{l}$, then $G$ is not $d$-flattenable.
\end{observation}

We use this to show that the ``banana'' graph in $5$ vertices is not 2-flattenable for $l_1$ (and $l_\infty$) in Theorem \ref{thm:banana}. The banana is a $K_5$ graph with one edge removed.

\section{$l_2$: Flattenability, Genericity, Independence in Rigidity Matroid}
\label{meera}

In this section we show relationships between $d$-flattenability and combinatorial
rigidity concepts via the cone $\cone{n}{l_p}$.

The definition of $d$-flattenability of a graph $G$ in $l_p$ requires 
every $l_p$ framework of the graph $G$ -- in an arbitrary dimension -- 
to be $d$-flattenable.
   
To accommodate the arbitrary dimension of the original framework, we first give  a suitable definition 
of  generic frameworks for $d$-flattenability.
\begin{definition}
\label{def:generic-d-flat}
Given an $l_p$ framework $(G,r)$, with $n$ vertices,  in arbitrary dimension, 
consider its pairwise length vector, $\delta_r$, in the cone $\cone{n}{l_p}$ (this was used in Section \ref{joel1}).
A framework $(G,r)$ of $n$ vertices is \define{generic}   
with respect to $d$-flattenability if the following hold:  
(i)
there is an open neighborhood $\Omega$ of $\delta_r$ in the (interior of the) cone $\cone{n}{l_p}$, 
(recalling that $n_p$ is the flattening dimension of  the complete graph $K_n$, 
$\Omega$ corresponds to an open neighborhood of $n_p$-dimensional point-configurations of $r$);
and
(ii)
$(G,r)$ is $d$-flattenable if and only if 
all the frameworks in $\Omega$ are.
\end{definition}

Item (i) implies that there is a ``full measure" or $n_p$-dimensional neighborhood of $(G,r)$ that corresponds to a neighborhood of distance vectors in the interior of the cone. Item (ii) asserts that all frameworks in this neighborhood are $d$-flattenable iff $(G,r)$ is.

\begin{theorem}
\label{thm:allgenericall}
  Every generic framework of $G$ is $d$-flattenable if and only if $G$ is $d$-flattenable.
\end{theorem}

\begin{proof}
The ``if'' direction follows immediately from the definition of
$d$-flattenability.
For the ``only if'' direction, notice that a non-generic, (bounded)
framework $(G,r)$ is a limit of a sequence $Q$ of
generic, bounded frameworks $\{(G,r_i)\}_i$, with a corresponding
sequence of pairwise distance vectors in $\cone{G}{l_p}$, and further a
corresponding sequence of projections onto the edges of $G$, i.e, a
sequence $Q'$ of bounded linkages
of $G$.
Because each $(G,r_i)$ is $d$-flattenable, each linkage in $Q'$ must be
realizable as  some
generic bounded framework $(G,r_i')$ in $d$-dimensions, i.e, each linkage
is the projection of the pairwise distance vector of some
$d$-dimensional bounded framework $(G,r')$, i.e, a pairwise distance
vector in the $d$-dimensional stratum of the cone $\cone{G}{l_p}$.
The projection of the limit framework $(G,r)$ of the sequence $Q$ is the
limit linkage of the
projected sequence $Q'$ of linkages with bounded edge lengths, whose
corresponding sequence of realizations as $d$-dimensional bounded
frameworks has a limit $(G,r')$. The latter limit follows from the fact
that the realization map that takes a linkage to its $d$-dimensional
bounded framework realizations (for any given bound) is a closed map.
This completes the proof.
\qed
\end{proof}

A similar argument can be used to show that the projection of
$\conestrat{n}{l_p}{d}$ onto the edges of any graph $G$, denoted
$\conestrat{G}{l_p}{d}$ is closed. This result was shown for $l_2$ in
\cite{Tarazaga}.

Although $d$-flattenability is equivalent to the presence of an inherent convex Cayley configuration space for $G$,
(as shown in Section \ref{joel1}), 
     we now move beyond inherent convex Cayley configuration
    spaces to  Cayley configuration spaces 
    over specified non-edges 
    $F$. These could be convex even if $G$ itself is not  $d$-flattenable (simple examples can be found for $d=2,3$ for $l_2$ in \cite{SiGa:2010}).
    A complete characterization of such $G,F$ is shown in \cite{SiGa:2010}, in the case of $l_2$ norm for $d=2$,
    conjectured for $d = 3$, and completely open for $d>3$.
    In Section \ref{open}, we extend the conjecture for general $d$.

An analogous theorem to Theorem \ref{thm:allgenericall} can be proven for the property 
of a $d$-dimensional framework $(G,r)$ having a convex Cayley configuration space
over specified non-edge set $F$.
However, since this framework is $d$-dimensional rather than of arbitrary dimension,
the definition of genericity has to be modified from Definition \ref{def:generic-d-flat}.

\begin{definition}
\label{def:generic-cayley}
Let $\delta_r$ be as in Definition \ref{def:generic-d-flat}.
A framework $(G,r)$ of $n$ vertices in $d$-dimensions is \define{generic}   
with respect to the property of convexity of $\Phi_{F,l_p}^{d}({G,\delta^G})$ if 
(i)
there is an open neighborhood $\Omega$ of $\delta_r$ in the stratum $\conestrat{n}{l_p}{d}$, 
(this corresponds to an open neighborhood of $d$-dimensional point-configurations of $r$);
and
(ii)
$(G,r)$  has convex Cayley configuration space over $F$ if and only if 
all the frameworks in $\Omega$ do.
\end{definition}

\begin{theorem}
\label{thm:allgenericcayley}
 Every generic $d$-dimensional framework $(G,r)$ 
has a convex Cayley configuration space over $F$ if and only if 
for all $\delta^G$, the linkage $(G,\delta^G)$ has a $d$-dimensional, convex Cayley configuration space over $F$.
\end{theorem}

\begin{proof}
The ``if'' direction follows immediately from the definitions.
Moreover, it is sufficient to prove the ``only if'' direction
for edge length vectors $\delta_G$ that are attained by some (potentially non-generic) $d$-dimensional framework $(G,r)$, 
because otherwise
the $d$-dimensional Cayley configuration space of the linkage $(G,\delta^G)$ is empty and hence trivially convex.
Now as in Theorem \ref{thm:allgenericall}, every non-generic $d$-dimensional framework $(G,r)$  
with edge length vector $\delta^G$ 
is a limit of a sequence $\{(G,r_i)\}_i$ of generic frameworks with edge length vectors $\delta^{G,i}$.
Since
convexity of  the Cayley configuration space 
$\Phi_{F,l_p}^{d}{G,\delta^{G,i}}$ is preserved over a open neighborhoods of $(G,r_i)$, it follows that
the limit of the sequence of spaces
$\{\Phi_{F,l_p}^{d}({G,\delta^{G,i}})\}_i$ exists, is convex, and is the  
Cayley configuration space 
of $(G,r)$. Since $(G,r)$ was chosen to have the edge length vector $\delta_G$, this space is 
in fact
$\Phi_{F,l_p}^{d}({G,\delta^{G}})$, 
the Cayley configuration space of the linkage $(G,\delta_G)$.
\qed
\end{proof}

A {\em property} of frameworks is said to be {\em generic}
if the {\em existence} of a generic framework with the property implies that the property holds for {\em all} generic frameworks.
Next we show that neither of the properties discussed above is a generic property of frameworks even for $l_2$.

\begin{theorem}
$d$-flattenability and convexity of Cayley configuration spaces over specific non-edges $F$ are not a generic property of frameworks
$(G,r)$. 
\end{theorem}

\begin{proof}
For $d$-flattenability:
    since the flattening dimension $n_2$ of $K_n$ in $l_2$ is $n-1$,  we show the counterexample of a $5$-vertex 
graph $G$ for which one
generic $4$-dimensional framework $(G,r)$ and its  neighborhood is $2$-flattenable in $l_2$,
while another such neighborhood is not.
See Figure \ref{fig:notgenericprop}.
For convexity of Cayley configuration spaces:
there are minimal, so-called  Henneberg-I graphs \cite{SiWa:2013} 
$G$, constructed on a {\em base} or initial edge $f$ with the following property:
for some 2-dimensional frameworks (and neighborhoods) $(G,r)$ with edge length vector $\delta^G$,
the 1-dimensional Cayley configuration space $\Phi_{f,l_2}^2(G\setminus f, \delta^{G\setminus f})$
(i.e, the attainable lengths for $f$)
is a single interval, while for other such frameworks (and neighborhoods) it is 2 intervals.
Please see Appendix in \cite{SiWa:2013}.
\qed
\end{proof}

\begin{figure}
  \centering\tiny
  \def\svgwidth{0.99\linewidth} 
  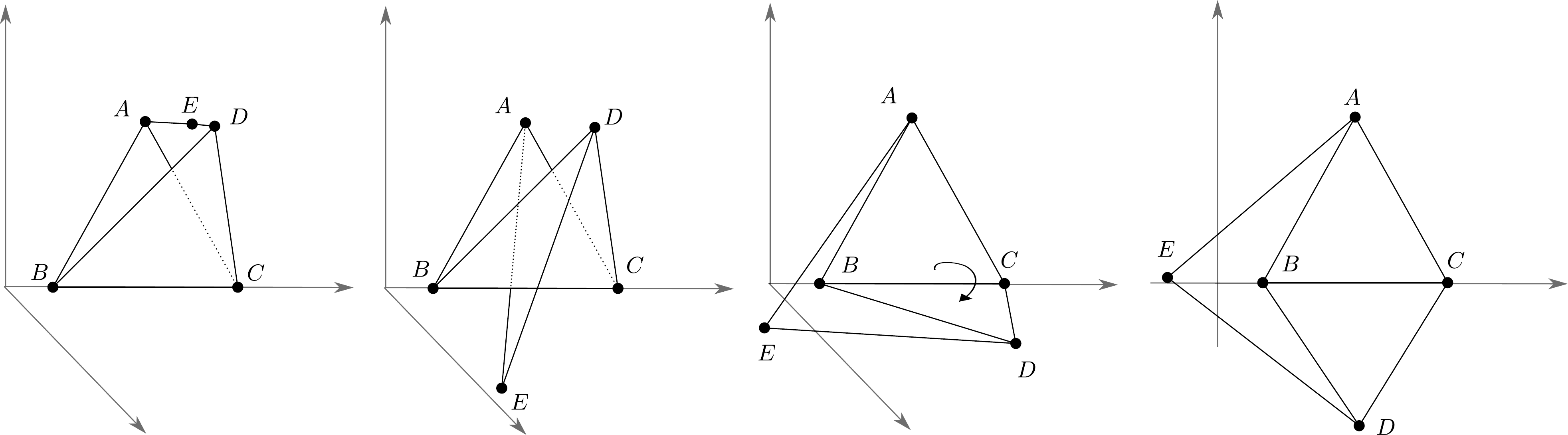 
  \caption{2 realizations of the same graph. In the first figure (left), we have edge lengths for $(a,e)$ and $(d,e)$ that do not allow $G$ to be flattened. The second graph is realized in 3-dimensions, but by ``unfolding it'' as shown, we can flatten it into 2-dimensions}
\label{fig:notgenericprop}
\end{figure}

Next, we consider the implication of the {\em existence} of a generic $d$-flattenable framework. Specifically, we prove two theorems connecting the $d$-flattenability with independence in the rigidity matroid:  we use the notion of rigidity matrix, and consequently regular frameworks and generic rigidity matroid developed by Kitson \cite{Kitson:2014}, as well as the equivalence of finite and infinitesimal rigidity using the notion of  {\em well-positioned} frameworks, which  
intuitively means that the $l_p$ balls of size given by the corresponding edge-lengths centered at the points
intersect properly (i.e, the intersection of $k$ $(d-1)$-dimensional ball boundaries is of dimension $d-k$).

The ``if'' direction of this next theorem is a restatement of Proposition 2 in Asimow and Roth \cite{MR0511410}. We extend their result here and show the other direction as well.

\begin{theorem}
\label{thm:genericimpliesindep}
For general $l_p$ norms, there exists a generic $d$-flattenable framework of $G$ if and only if $G$ is independent in the $d$-dimensional generic rigidity matroid.
\end{theorem}

\begin{figure}
  \centering\tiny
  \def\svgwidth{0.35\linewidth} 
  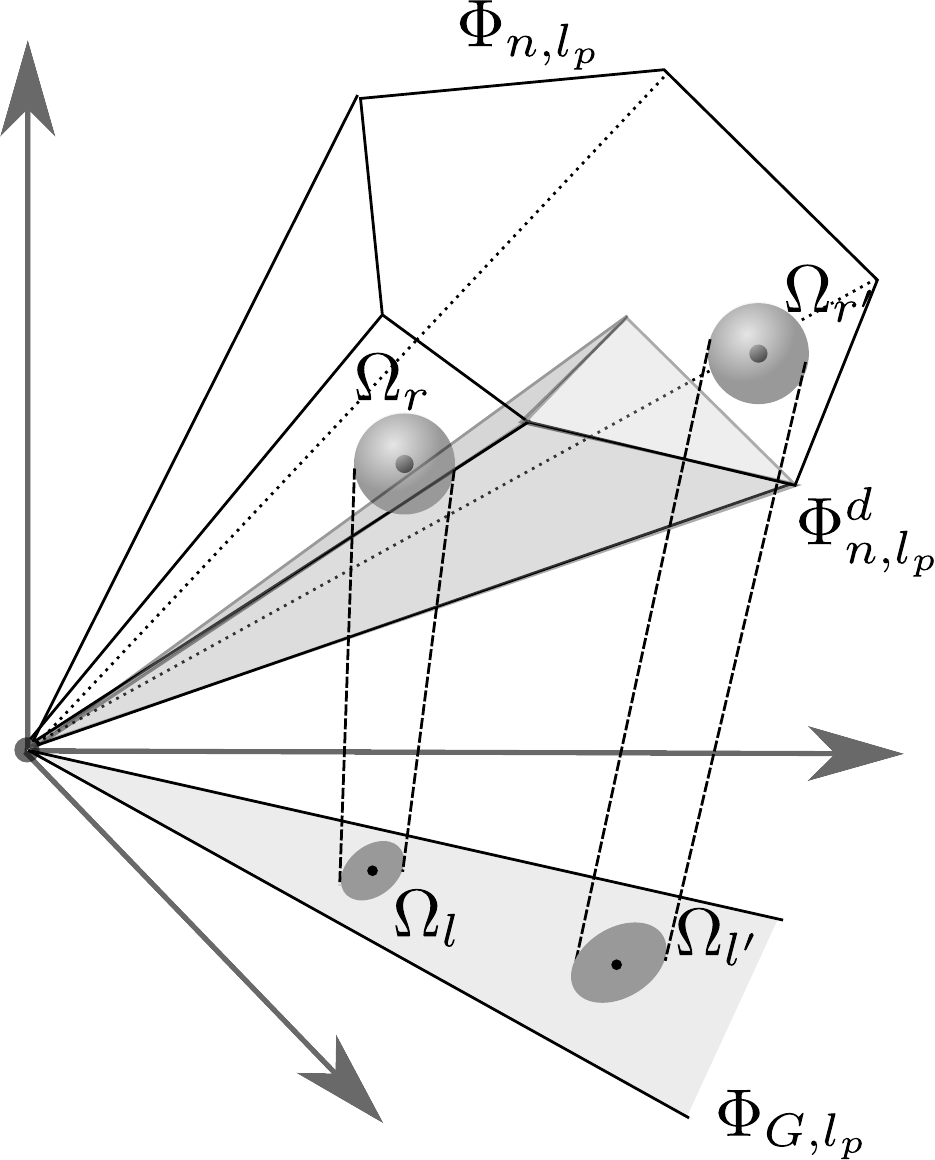 
  \def\svgwidth{0.35\linewidth} 
  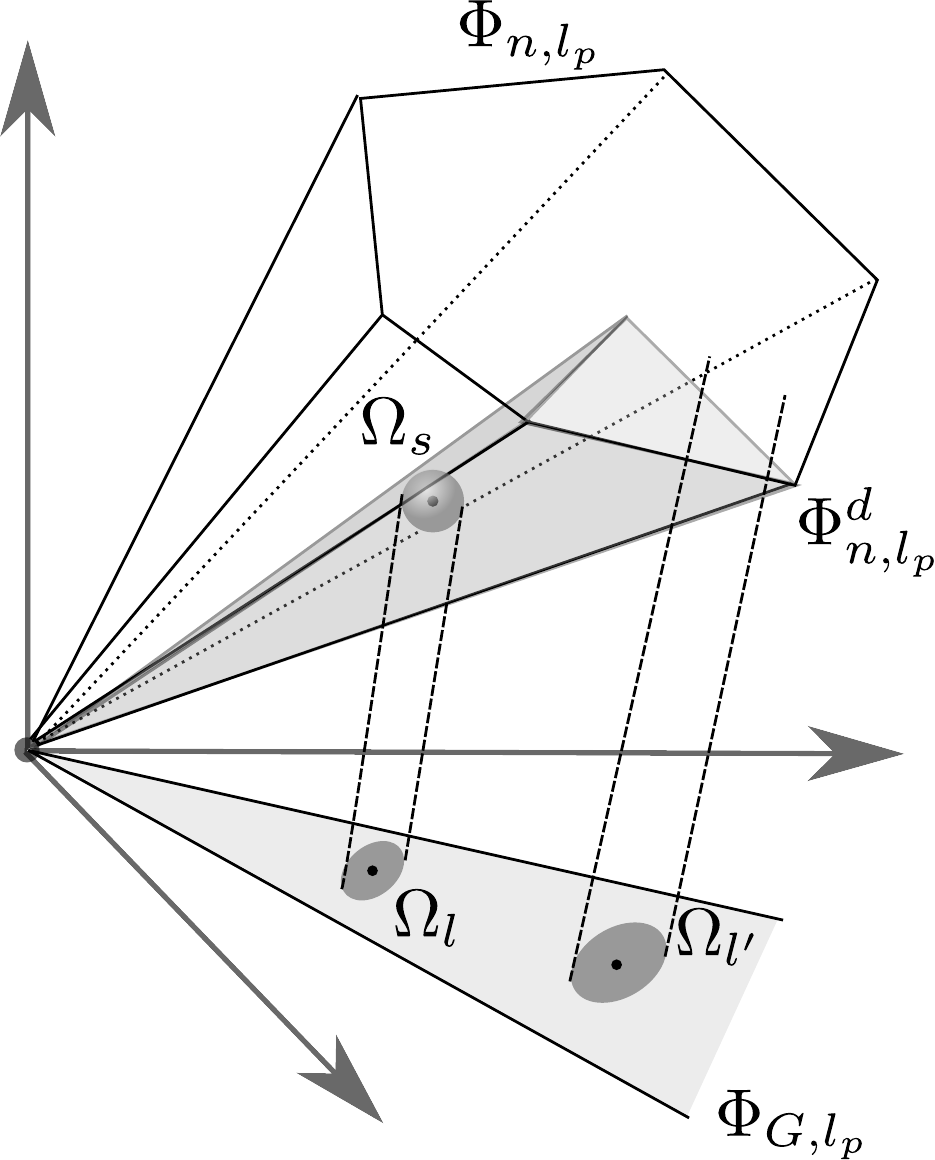 
  \caption{On the left we have 2 neighborhoods $\Omega_r$ and $\Omega_{r'}$ of 2 distance vectors $\delta_r$ and $\delta_{r'}$ in the cone. We then project $\Omega_r$ and $\Omega_{r'}$ onto the edges of $G$ to obtain $\Omega_l$ and $\Omega_{l'}$, which are essentially the neighborhoods of $(G, \delta_r^G)$ and $(G, \delta_{r'}^G)$. On the right, we then take the fiber of $\Omega_l$ and $\Omega_{l'}$ on $\conestrat{n}{l_p}{d}$. The fiber of $\Omega_l$ is completely contained in the stratum while that of $\Omega_{l'}$ misses (does not intersect) the stratum.}
\label{fig:neighborhoods}
\end{figure}

\begin{proof}

For the forward direction, we  note that existence of a generic $d$-flattenable framework $(G,r)$ is equivalent to the statement that the pairwise distance vector $\delta_r$ has an open neighborhood $\Omega_r$ in the interior of the cone $\cone{n}{l_p}$,  
and the  $d$-flattenings $(G,s)$  form an open neighborhood $\Omega_s$ of pairwise distance vector $\delta_s$ in the relative interior of the stratum $\conestrat{n}{l_p}{d}$.  This is also equivalent to saying there is a corresponding open neighborhood of $d$-flattenable linkages $(G,\delta_s^G= \delta_r^G)$. Now $\Omega_s$ (resp. $\Omega_r$) must contain an open neighborhood of pairwise distance vectors $\delta_s$   (resp. $\delta_r$) that correspond to well-positioned and regular frameworks $(G,s)$, (resp. $(G,r)$), hence without loss of generality, we can take that neighborhood to be 
$\Omega_s$ (resp. $\Omega_r$), consisting of $d$-dimensional, well-positioned, regular frameworks $(G,s)$ (resp. $(G,r)$)  that are realizations of an open neighborhood of $\Omega_G$ of 
linkages $(G, \delta_s^G = \delta_r^G)$. These linkages correspond to a coordinate shadow or projection of $\Omega_r$ and $\Omega_s$ onto (the edges in) $G$.

Now observe that  the generic rigidity matrix of $G$ is  the Jacobian of the distance map from the $d$-dimensional point-configuration $s$ to the edge-length vector $\delta_s^G$ at the point $s$.
For $l_2$ and integral $p > 1$, this map is clearly specified by polynomials. For $l_1$ and $l_\infty$, we use the notion of well-positioned frameworks from \cite{Kitson:2014}. If the frameworks of $\Omega_r$ are well-positioned, then it follows that the distance map is locally specified by linear polynomials corresponding to a relevant facet of the $l_1$ or $l_\infty$ ball. Because $\Omega_r$ has dimension equal to the number of edges in $G$, these polynomials are algebraically independent. Hence, their Jacobian has rank equal to the number of edges in $G$.
Therefore, the existence of well-positioned, regular realizations $s$ to an entire neighborhood 
of edge-length vectors $\delta_s^G$ implies the statement that the rows of the generic rigidity matrix -- that correspond to the edges of $G$ -- are independent. 

The converse follows from Proposition 2 of Asimow-Roth \cite{MR0511410} observing that at well-positioned and regular points, for all $1\le p \le \infty$ the $l_p$ distance map from point configurations to 
pairwise distance vectors is a smooth map.
\qed
\end{proof}

The following corollary is immediate from the forward  direction of the above proof.

\begin{corollary}
For general $l_p$ norms,
a graph $G$ is $d$-flattenable only if $G$ is independent in the $d$-dimensional rigidity matroid.
\end {corollary}

The following theorem and corollary utilize the dimension of the projection of the $d$-dimensional stratum on the edges of $G$ from the above proof. Note that in the above proof, if $G$ is an $n$-vertex graph, the neighborhood $\Omega_r$ has dimension $n_p$, i.e, the flattening dimension of $K_n$; $\Omega_s$ has dimension equal to that of the stratum $\conestrat{n}{l_p}{d}$, and   $\Omega_G$ has dimension equal to the number of edges of $G$ (see Figure \ref{fig:dimensiondifferences}).

The first item of the following Theorem is a restatement of Theorem \ref{thm:genericimpliesindep} and as such, the ``only if'' direction appears in \cite{MR0511410}.

   \begin{theorem}
   \label{thm:bigtheorem}
For general $l_p$ norms, a graph $G$ is 
\begin{enumerate}
    \item 
        independent in the generic $d$-dimensional rigidity matroid (i.e, the rigidity matrix of a well-positioned and regular framework has independent rows), if and only if coordinate projection of the stratum $\conestrat{n}{l_p}{d}$ onto $G$ has dimension  equal to the number of edges of $G$;
        \label{item:1}
    \item
    maximal independent (minimally rigid) if and only if projection of the stratum $\conestrat{n}{l_p}{d}$ onto $G$ is maximal (i.e., projection preserves dimension) and is equal to the number of edges of $G$;
    \item 
    rigid in $d$-dimensions if and only if projection of the stratum $\conestrat{n}{l_p}{d}$ onto $G$ preserves its dimension;
    \item
    not independent and not rigid in the generic $d$-dimensional rigidity matroid if and only if the projection of 
    $\conestrat{n}{l_p}{d}$ onto $G$ is
        strictly smaller than the minimum of: the dimension of the stratum and the number of edges
        in G.
\end{enumerate}

\end{theorem}

\begin{proof}
  The proof of this theorem follows from the proof of the previous result: Theorem \ref{thm:genericimpliesindep}. Each case is illustrated in Figure \ref{fig:dimensiondifferences}.
\qed
\end{proof}

\begin{figure}
  \centering
  \def\svgwidth{0.22\linewidth} 
  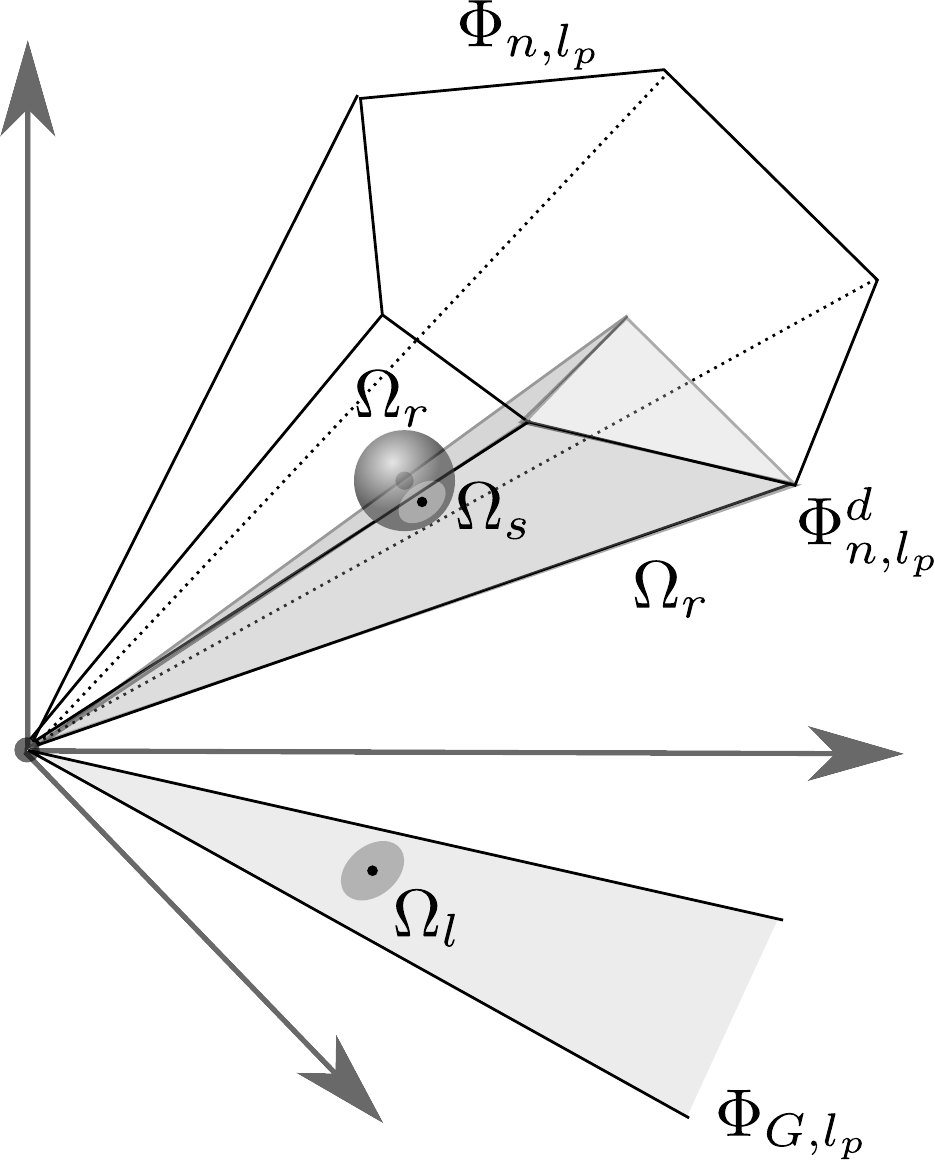 
  \hspace{5pt}
  \def\svgwidth{0.22\linewidth} 
  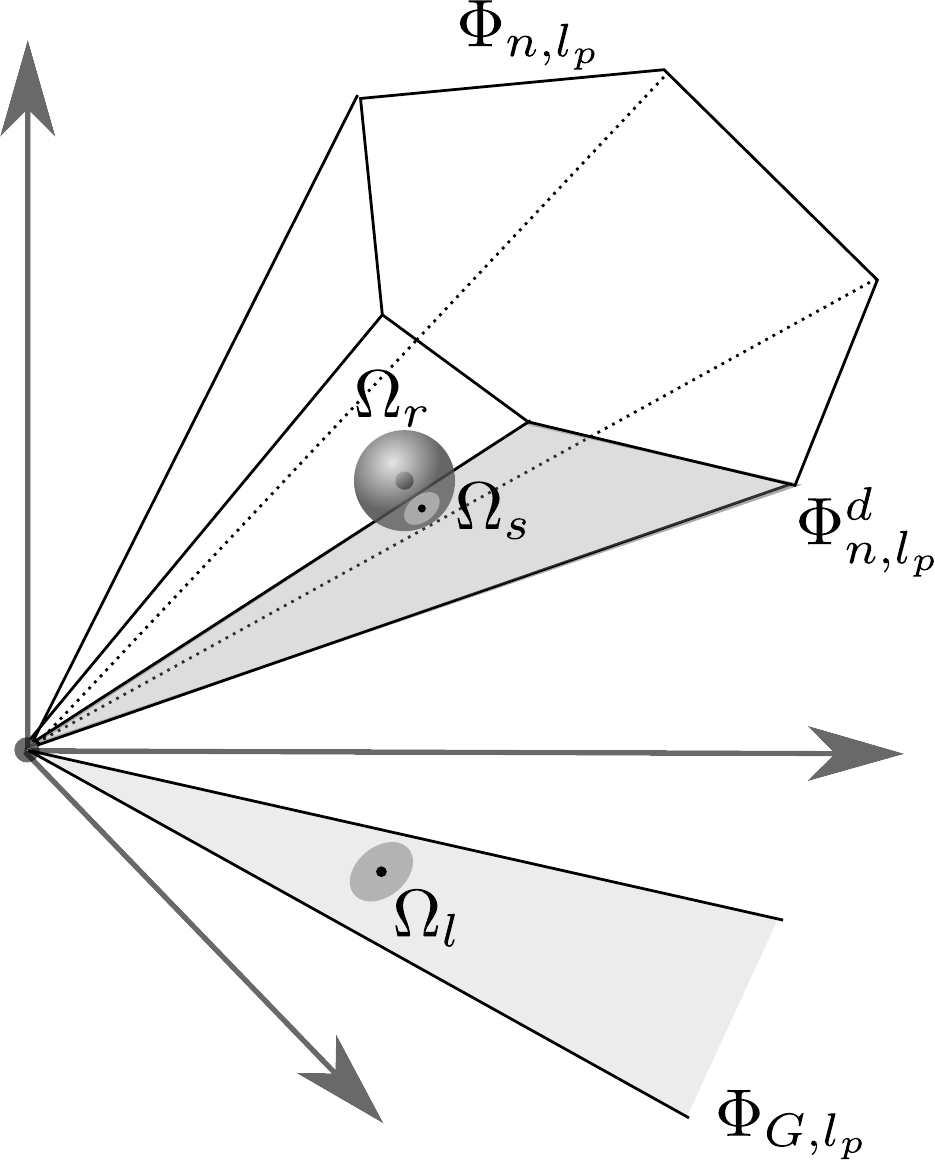 
  \hspace{5pt}
  \def\svgwidth{0.22\linewidth} 
  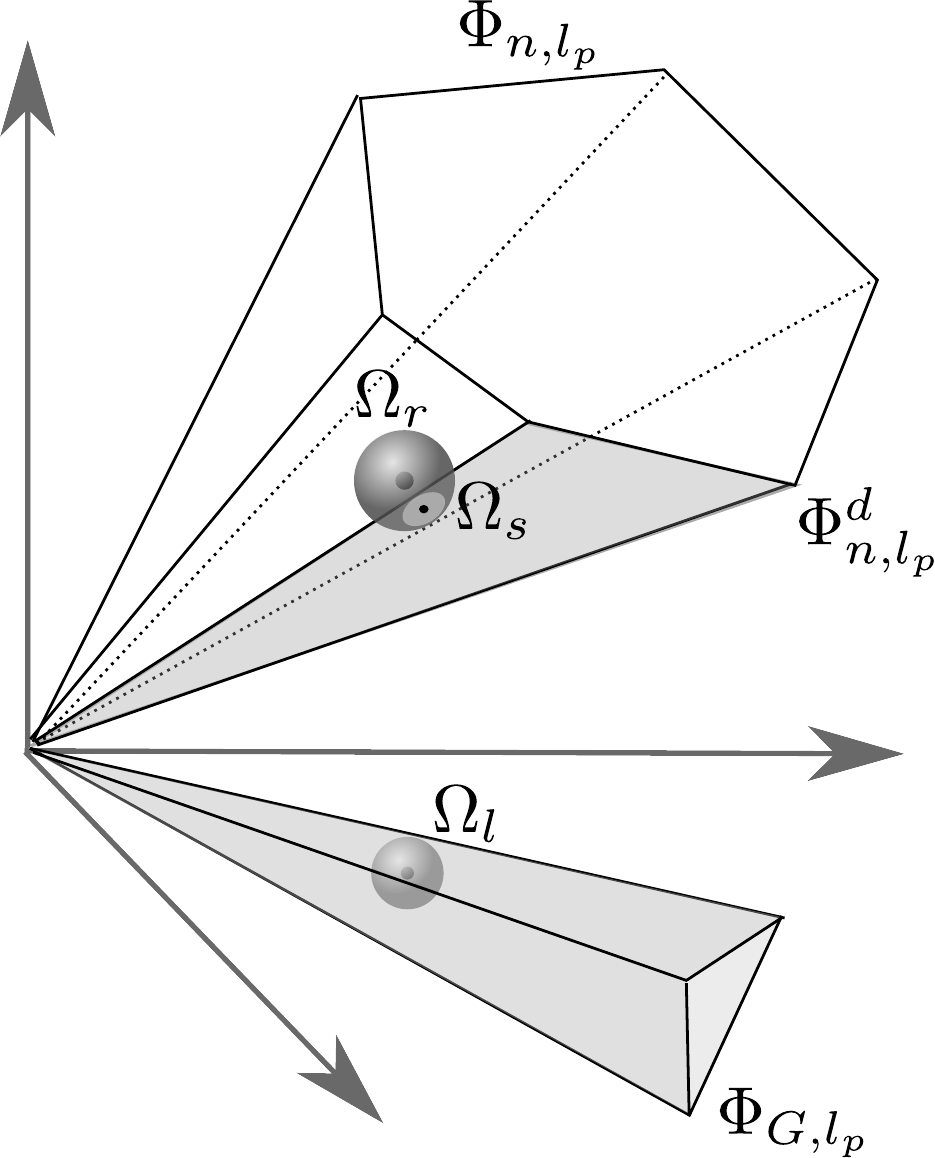 
  \hspace{5pt}
  \def\svgwidth{0.22\linewidth} 
  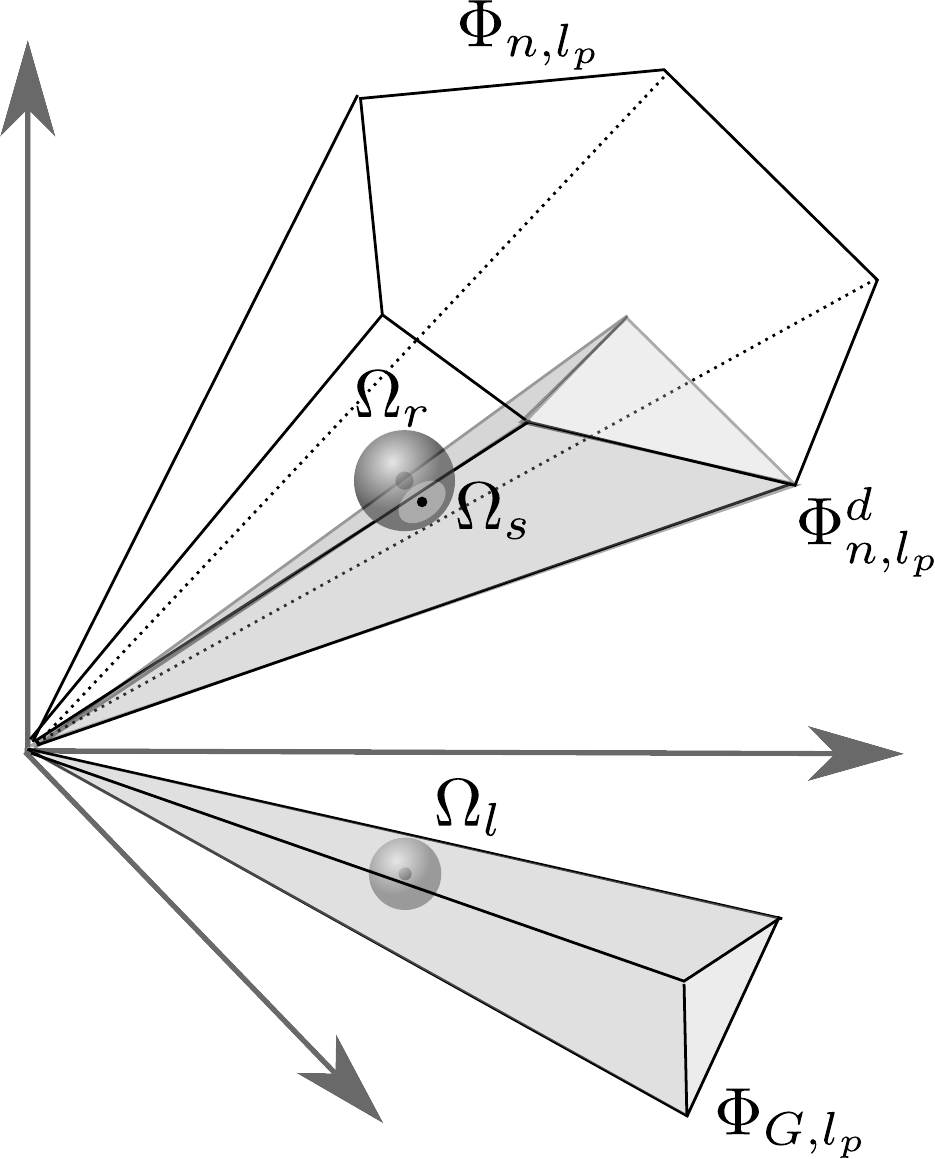 
  \caption{These are visualizations of when frameworks are isostatic and independent. In all of these cases $dim(\Omega_r) \ge \max\{dim(\Omega_s,dim(\Omega_l))\}$. We only show 2 and 3 dimensions here, but in general the dimensions will be much higher. See Figure \ref{fig:neighborhoods} for explanation of what each is. In the following, when we use equality or inequality, we are referring to dimension. On the left, $\Omega_s = \Omega_l < \conestrat{n}{l_p}{d}$ meaning $\delta_r$ is independent but not isostatic. Middle left: $\Omega_s = \Omega_l = \conestrat{n}{l_p}{d}$, so $\delta_r$ is maximal independent or isostatic. Middle right: $\Omega_s = \conestrat{n}{l_p}{d} < \Omega_l$ meaning $\delta_r$ is rigid but not independent. Right: $\Omega_s < \Omega_l$ and $\Omega_s < \conestrat{n}{l_p}{d}$ meaning $\delta_r$ is neither independent nor rigid.}
\label{fig:dimensiondifferences}
\end{figure}

We note that the ``only if'' direction of item \ref{item:1} of Theorem \ref{thm:bigtheorem} is the same result that appears as Proposition 2 (in a different form) in \cite{MR0511410}. However, to the best of our knowledge, the ``if'' direction and the rest of Theorem \ref{thm:bigtheorem} and the proof of Theorem \ref{thm:genericimpliesindep} are new results. We obtain the following useful corollary.

\begin{corollary}
For $l_p$ norms, the rank of a graph $G$ in the $d$-dimensional rigidity
matroid is equal to the dimension of the projection
$\Phi_{G,l_p}^d$ on $G$ of the $d$-dimensional stratum $\Phi_{n,l_p}^d$.
\end{corollary}

    \section{$l_1$: 2-flattenability}
\label{joel2}
We now turn our attention to the $l_1$ norm in 2-dimensions. We note that the $l_1$ and $l_\infty$ norms in 2-dimensions are equivalent by simply applying a rotation to our axes (for argument, see \cite{DezaAndLaurent}). Specifically, we would like to characterize the class of graphs that are 2-flattenable under the $l_1$ norm. A result from \cite{Witsenhausen1986} shows that $K_4$ is 2-flattenable. We note that $K_4$ is the only forbidden minor for 2-flattenability under the $l_2$ norm. It immediately follows that the 2-flattenable $l_2$ graphs are a strict subset of the 2-flattenable $l_1$ graphs. In the remainder of this section, we narrow down the possible candidates for forbidden minors.

\begin{observation}
  All partial 2-trees are 2-flattenable.
\end{observation}

This follows from the fact that partial 2-trees are exactly graphs without a $K_4$ minor. We define 2-trees recursively. A triangle is a 2-tree. Given any 2-tree, attaching another triangle onto a single edge is also a 2-tree. A partial 2-tree is any subgraph of a 2-tree. Because the 2-flattenable graphs for $l_2$ are exactly the partial 2-trees, it follows partial 2-trees are 2-flattenable for $l_1$.

In order to generalize our results, we introduce the following Theorem which involves a \define{2-sum} operation. A \define{2-sum} of graph $G_1$ and $G_2$ is a a new graph $G$ made by gluing an edge of $G_1$ to one of $G_2$, i.e. we identify an edge of $G_1$ with an edge of $G_2$.

\begin{theorem}\label{thm:2sumK4}
    A 2-sum of 2-flattenable graphs is 2-flattenable if and only if at most one graph has a $K_4$ minor.
\end{theorem}

\begin{figure} 
    \centering 
    \def\svgwidth{100pt} 
    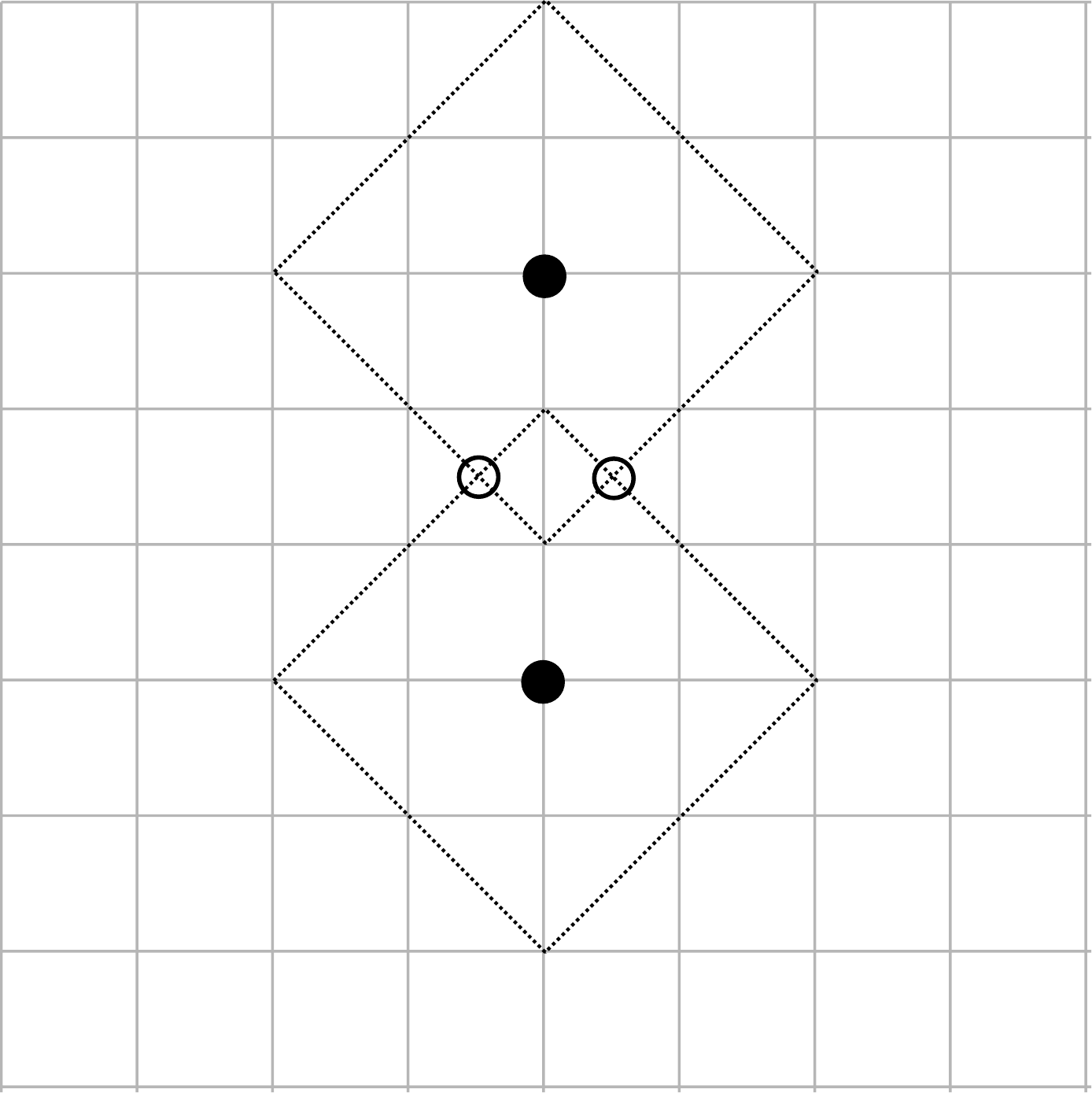 
    \def\svgwidth{100pt} 
    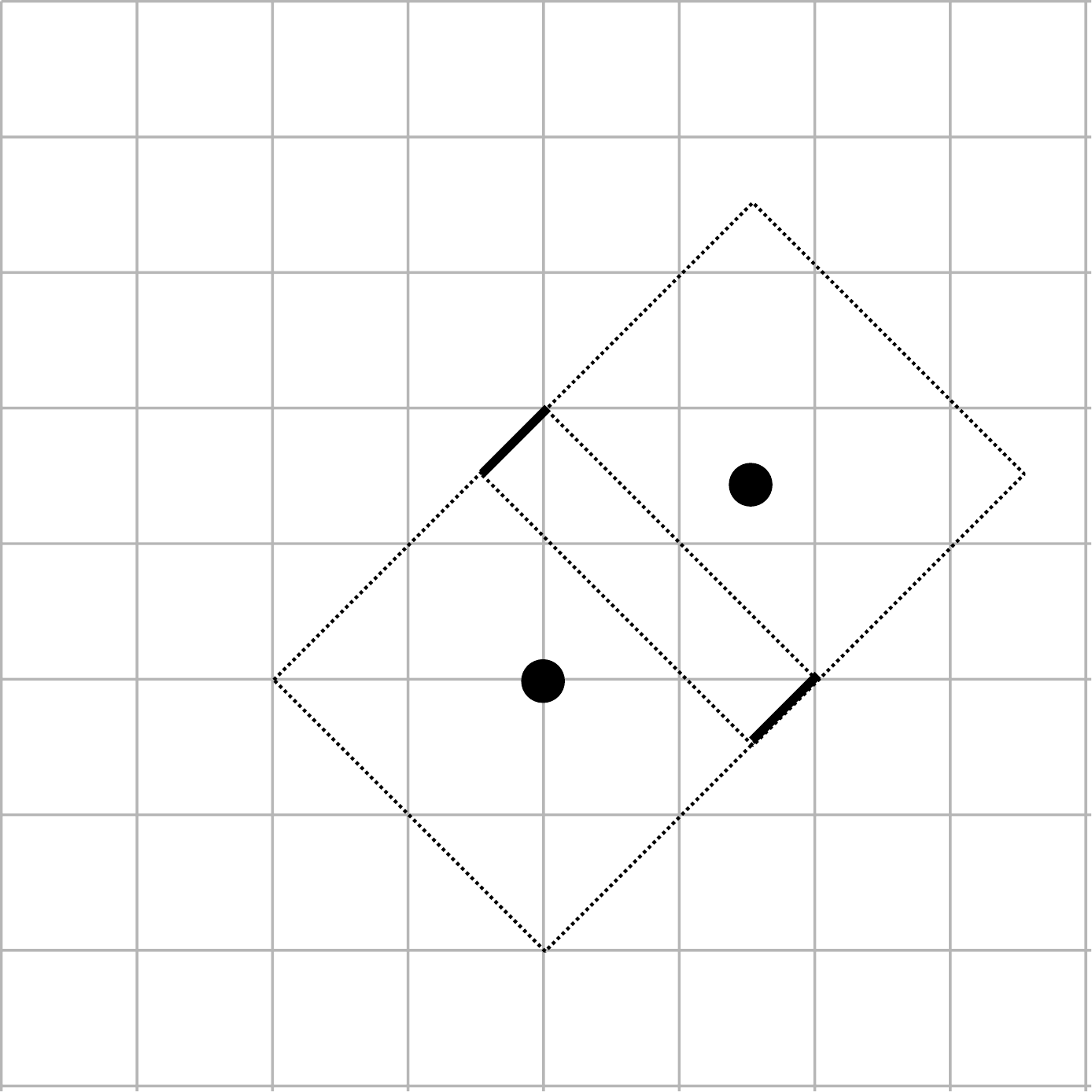
    \caption{On the left is a partial realization of $G_2$ if we assume a vertical orientation for $(v_1, v_2$). On the right is the same for $(v_1, v_2)$ at an angle of 45 degrees}
\label{fig:2sumproof}
\end{figure}

\begin{proof}
    Suppose $G_1$ and $G_2$ are 2-flattenable and only $G_1$ has a $K_4$ minor. Then, $G_2$ is a partial 2-tree. Thus the 2-sum of $G_1$ and $G_2$ can be built by taking a realization of $G_1$, identifying the 2-sum, and then adding the vertices of $G_2$ one at a time. Let $r$ and $s$ be the 2 vertices we are attaching some new vertex $v$ to. No matter the orientation of $r$ and $s$, as long as the triangle formed by $r, s, \text{and } v$ obeys the triangle inequality, the $l_1$-balls surrounding $r$ and $s$ with distances corresponding to their distance to $v$ will always intersect in 2-dimensions. This can be verified by placing $r$ at the origin, moving $s$ along the $l_1$ ball of $r$ in the first quadrant and observing the balls surrounding $r$ and $s$ as $s$ moves. Hence, the triangle $r,s,v$ can be realized in 2-dimensions.

    Suppose $G_1$ and $G_2$ both have a $K_4$ minor. We give a counter-example to show that the 2-sum is not 2-flattenable. Let $G_1$ be the equidistant $K_4$ with each edge having distance 3. Let $G_2$ have every edge with distance 2 except for $(v_1, v_2)$, which has distance 3. $G_1$ has only one realization modulo rearranging vertices: all points at the corners of the distance 3 $l_1$-ball. The edges of $G_1$ are all either vertical/horizontal or at an angle of 45 degrees. We claim $G_2$ has no realization with $(v_1,v_2)$ at those angles. 

    If we assume that $(v_1,v_2)$ is vertical, looking at figure \ref{fig:2sumproof}, we see that the remaining 2 vertices can only lie at $p_1$ and $p_2$. The possible distances they can obtain are 0 and 1, which means $G_2$ cannot be completed. Looking at the 45 degree case on the right of figure \ref{fig:2sumproof}, we see that the other 2 vertices can only lie in $I_1$ and $I_2$. This leads to possible distances of $[0,1]$ and 4. Thus $G_2$ still cannot be completed. Note that the horizontal and other 45 degree orientations are just flips of these two cases.

    Hence, $G_2$ has no realization with $(v_1,v_2)$ at any of the angles of $G_1$'s edges. So, the 2-sum of $G_1$ and $G_2$ is not 2-realizable. We note that this 2-sum does have a realization in 3-dimensions: $v_1 = (0,0,0), v_2=(1.5,1.5,0), v_3=(0.5,1,0.5), v_4=(1,0.5,-0.5)$ give a realization for $G_2$ with $(v_1,v_2)$ at a 45 degree angle, so $G_1$ 2-sum $G_2$ is not 2-flattenable.
\qed
\end{proof}

Note that the realization given at the end of the above proof has a $K_4$ lying on a 2-dimensional subspace of $\R^3$. This is still a 3-dimensional realization as the spanned 2-dimensional subspace is not equipped with an 2-dimensional $l_1$ norm. However, it is equipped with an $l_1$ norm in 3-dimensions. The $K_4$ has a realization in $R^2$, we only need it in 3-dimensions in the proof to 2-sum it to the other $K_4$.

Another result from \cite{Ball} shows that $K_5$ is not 2-flattenable. Hence, we search the subgraphs of $K_5$ and check them for 2-flattenability. This leads to the following example of a non-2-flattenable graph, which we prove using the techniques developed in this paper.

\begin{theorem}\label{thm:banana}
  The so called ``banana'' graph or $K_5$ minus one edge is not 2-flattenable under the $l_1$ norm. 
\end{theorem} 

\begin{figure} 
    \centering 
    \def\svgwidth{100pt} 
    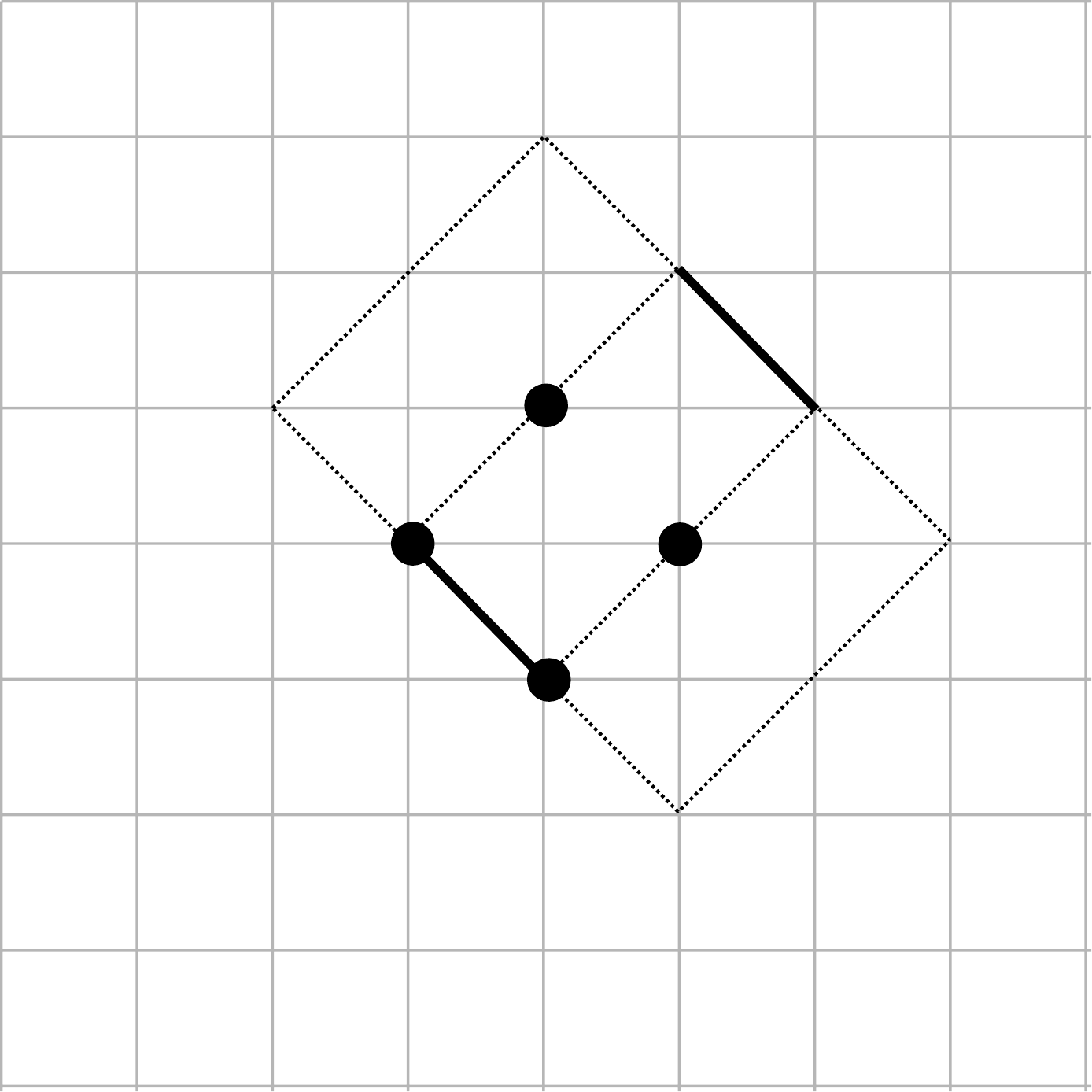 
    \def\svgwidth{100pt} 
    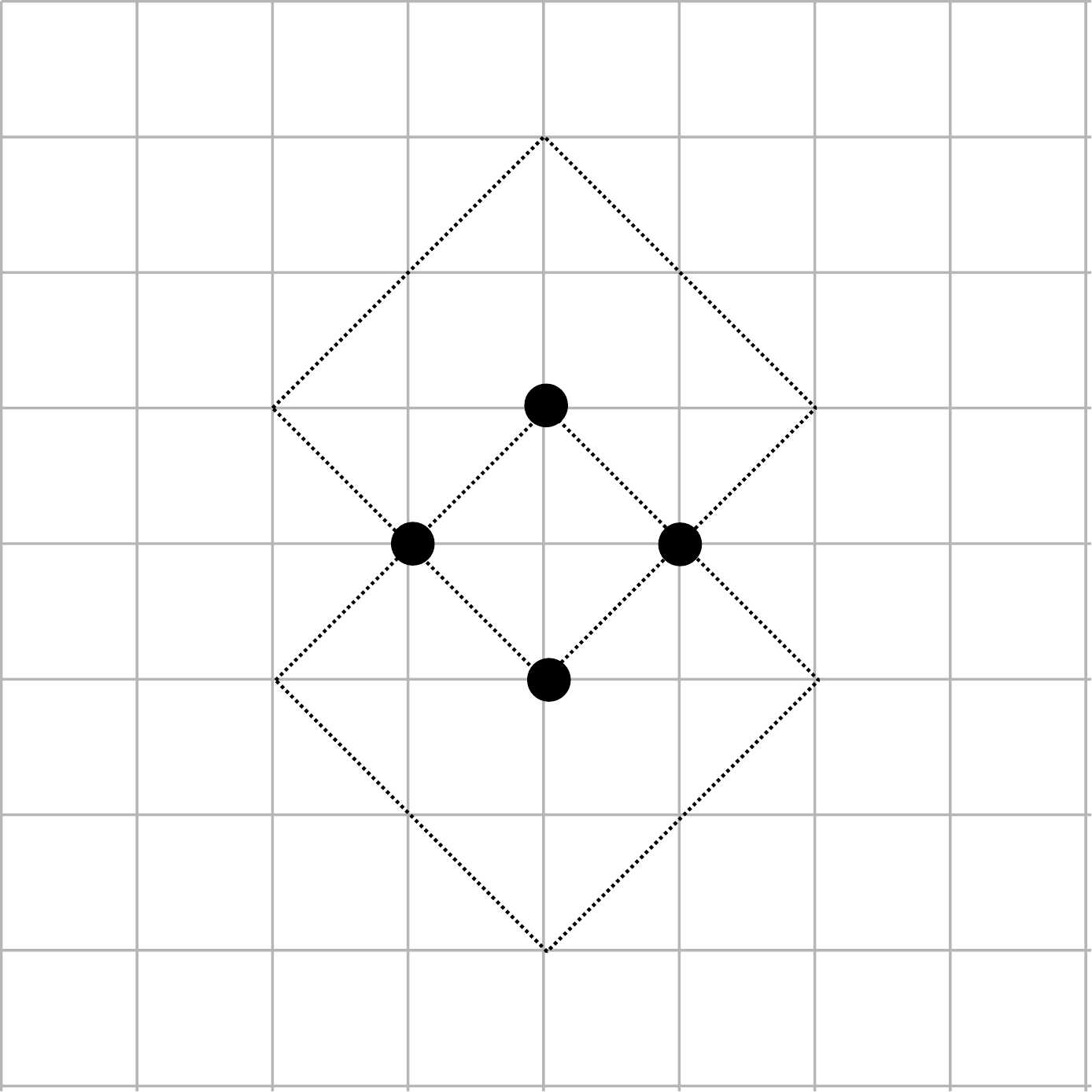
    \caption{On the left is an illustration of an equidistant $K_4$ with the possible intervals of Case 1. On the right is the same for Case 2}
\label{fig:K_4proof}
\end{figure}

\begin{proof}
  We will invoke Observation \ref{obs:not_convex_not_flattenable} to show this.

  Consider a distance vector for the banana with unit distances for all except one edge, $f$. This has a realization in 3-dimensions as $K_5$ is 3-flattenable for the $l_1$ norm (see \cite{Ball}). Then, we have an equidistant $K_4$ as a subgraph. The only realization for such a $K_4$ in 2-dimensions is to have all 4 points arranged as the vertices of the unit ball centered at the origin. The 2 remaining unit edges then connect a new vertex to 2 of these points. Here we have 2 cases: the 2 vertices border the same quadrant or they lie across one of the axes from each other.

  Case 1: Without loss of generality, we assume the 2 vertices are the upper right of the $K_4$. In figure \ref{fig:K_4proof}, it can be seen that the new vertex can lie anywhere in $I_1$ or $I_2$. If it lies in $I_1$, the remaining edge of the banana can take lengths in the range $[0,1]$. If it lies in $I_2$, the only length it can be is 2.

  Case 2: Without loss of generality, assume the 2 vertices are the top-most and bottom-most. Again from figure \ref{fig:K_4proof}, the new vertex only has 2 positions it can be in, each leading to a length of 1 for the remaining edge.

  Hence, $\cayley{F}{l_1}{2}{G \setminus F}{\delta^{G \setminus F}} = [0,1] \cup \{2\}$, where $G$ is the banana and $F=\{f\}$. This is not convex and thus by Theorem \ref{thm:covexcayley_d-flattenable}, the banana is not 2-flattenable.
\qed
\end{proof}

\begin{observation}\label{obs:K5minus2}
    $K_5$ minus 2 edges incident to a single vertex is 2-flattenable.
\end{observation}

\begin{proof}
    This follows directly from Theorem \ref{thm:2sumK4}.
\qed
\end{proof}

\begin{observation}
    Connected graphs on 5 vertices with 7 edges are 2-flattenable
\end{observation}

\begin{proof}
    The only 2 such graphs are a subgraph of the Observation \ref{obs:K5minus2} and the complete 2-tree on 5 vertices. Both of these we know to be 2-flattenable. 
\qed
\end{proof}

The only remaining 5 vertex graph we have not looked at yet is the \define{wheel} graph. So far we have shown that for 5 vertices, graphs with the wheel as a minor are not 2-flattenable and graphs without are 2-flattenable. Thus, if we can show that the wheel is not 2-flattenable, then it becomes the only forbidden minor for $l_1$ 2-flattenability. We discuss this more in Section \ref{open} and conjecture that in fact the wheel is the only forbidden minor for 2-flattenability under the $l_1$ (and $l_\infty$) norm.

\section{Conjectures and Open Problems}
\label{open}

\subsection{Combinatorial Rigidity and Structure of $\cone{n}{l_p}$}
\label{sec:rigidityandstrucutre}

In Theorem \ref{thm:genericimpliesindep} and Theorem \ref{thm:bigtheorem}, we have
shown that combinatorial rigidity properties of a graph in $d$-dimensions
is tied to the dimension of the projection of {\sl some} "face" of  $d$-dimensional stratum of $\cone{n}{p}$. These properties are not generic
when viewed as properties of frameworks $\Omega_r$ in the flattening
dimension of $K_n$, i.e, when viewed as distance vectors $\delta_r$ in the
interior of the cone $\cone{n}{p}$. However, since we know that
these properties are generic in $d$-dimensions (via combinatorial rigidity
techniques), this means it must be that that the projection of \emph{every}
face of the $d$-dimensional stratum of $\cone{n}{p}$ onto $G$ has the same
dimension. Thus combinatorial rigidity and Cayley configuration spaces can
help understand the structure of the cone. However, it would be good to
have an independent proof of these properties directly via the cone geometry. More formally:

\begin{conjecture}
  $G$ is $d$-independent if and only if the projection of every face of $\conestrat{G}{l_p}{d}$ has dimension equal to the number of edges of $G$.
\end{conjecture}

It would be useful to show a stronger property about the continuous mapping used in the proof of Theorem \ref{thm:allgenericall}. Doing so would require deeper understanding of the $d$-flattening process itself. Formally:

\begin{question}
    \label{ques:cont_path}
    Given a realization $(G,r)$ of a $d$-flattenable linkage $(G,\delta_G)$ in some high dimension, is there always a continuous path from $(G,r)$ to $(G,r')$ in $d$-dimensions for general $l_p$ norms?
\end{question}

It would be useful to even show a weaker version of Question \ref{ques:cont_path}: Does a continuous path of high dimensional frameworks of a $d$-flattenable graph $G$ always correspond to a path in $d$-dimensional frameworks?

In the case of the Euclidean or $l_2$ norm many questions remain
concerning core results and applications of convex Cayley configuration
spaces.

The question of convexity of Cayley configuration spaces of graphs $G$
over specified edge sets $F$ is fully understood, and the proof
\cite{SiGa:2010} uses the existence of a specific type of homeomorphism
to produce forbidden minors. The property is relatively close to that of
$2$-flattenability which is equivalent to convexity of inherent
$2$-dimensional Cayley configuration spaces. In fact the class of graphs
(partial 2-trees) have convex Cayley configuration spaces in any dimension
(follows immediately from the close relationship to $2$-flattenability).
Thus, as  in Section \ref{sec:rigidityandstrucutre}, we expect that
fully understanding the structure of convex Cayley configuration spaces of
partial 2-trees in 2-dimensions (which relies on combinatorial rigidity
and forbidden minor properties) will help in understanding the structure of
$2$-dimensional stratum of the cone.

We believe the study of Cayley configuration spaces of partial 2-trees
can simplify results related to the so-called Walker conjecture about the
topology of Cartesian configuration spaces for a very simple class of
partial 2-trees, namely polygonal graphs \cite{Farber1,Farber2},
as well as to extend them to general, partial 2-trees. In fact, we believe
that the Cayley configuration space of partial 2-trees can help to understand 
entire structure of $\cone{n}{l_2}$.

While convex Cayley configuration spaces over specified non-edges $F$
in 2-dimensions are fully characterized, very little is known (beyond the
forbidden minors for $3$-flattenability) in higher dimensions. In
particular, there are graphs $G$ that are themselves not $3$-flattenable,
but their Cayley configuration spaces are convex over certain non-edges $F$. Several
natural conjectures in \cite{SiGa:2010} relate to the specific type of
homeomorphism used to produce the forbidden minor characterizations in the
2D case. These still remain open for higher dimensions.

\subsection{2-Flattenability under $l_1$}
In Section \ref{joel1}, we showed a number of techniques to prove (non)-2-flattenability of certain graphs under the $l_1$ norm. Mostly these dealt with a constructive argument like the partial 2 tree case to prove flattenability and showing non-convexity of inherent Cayley configuration space for non-flattenability. 

It is still an open question as to what the forbidden minor characterization of 2-flattenability under $l_1$ is. Our results show that the only 5 vertex graph to classify is the wheel. Due to the fact that the wheel is a minor to all of the other non-2-flattenable graphs on 5 vertices, we raise the following conjecture.

\begin{conjecture}
   The forbidden minor characterization for 2-flattenability under the $l_1$ and $l_\infty$ norms consists of only the wheel on 5 vertices.
\end{conjecture} 

Showing this requires only that we show that the wheel is not 2-flattenable. If this result is proven to be negative, then it will be necessary to look at 6 vertex graphs such as $K_{3,3}$, the doublet, and $K_{2,2,2}$.

\subsection{Other Metrics}
We would like to extend the results of this paper to other polyhedral norms faces.
Some of the major obstacles have been outlined in \cite{Kitson:2014}. In particular the 
nonexistence of well-positioned and regular frameworks, all of whose
sub-frameworks are also regular. Some work was done in this paper on this paper for the specific case of $l_1$. 

Extending the results of this paper to
other metrics would increase its applicability in combinatorial
optimization settings. Doing this will require us to first choose an appropriate
notion of dimension for metric topologies, be it the doubling dimension or some other classical notion of dimension.

\subsection*{Acknowledgement}
We thank Bob Connelly, Steven Gortler and Derek Kitson for
interesting conversations related to this paper.

\bibliography{refs}
\bibliographystyle{plain}

\end{document}